%% file: main.tex
\documentclass[acmsmall,screen,authorversion]{acmart}

\setcopyright{acmcopyright}
\copyrightyear{2023}
\acmYear{2023}
\acmDOI{}

\input{macros}

\begin{document}

\DontPrintSemicolon
\AlgoDontDisplayBlockMarkers
\SetAlgoNoEnd
\SetAlgoNoLine

\title[On Learning Polynomial Recursive Programs]{On Learning Polynomial Recursive Programs}

\author{Alex Buna-Marginean}
\email{alex.bunamarginean@spc.ox.ac.uk}
\affiliation{%
  \institution{Department of Computer Science, University of Oxford}
  \city{Oxford}
  \country{UK}} 

\author{Vincent Cheval}
\email{vincent.cheval@cs.ox.ac.uk}
\affiliation{%
  \institution{Department of Computer Science, University of Oxford}
  \city{Oxford}
  \country{UK}}

\author{Mahsa Shirmohammadi}
\email{mahsa@irif.fr}
\affiliation{%
  \institution{Universit\'e Paris Cité, CNRS, IRIF}
  \city{Paris}
  \country{France}}

\author{James Worrell}
\email{jbw@cs.ox.ac.uk}
\affiliation{%
  \institution{Department of Computer Science, University of Oxford}
  \city{Oxford}
  \country{UK}} 

\renewcommand{\shortauthors}{Buna-Marginean, Cheval, Shirmohammadi and Worrell}




\input{abstract}
\maketitle

\section{introduction}

\input{intro}

\section{Overview}
\label{sec:overview}

\input{overview}

\SetAlgorithmName{Figure}{Figure}{Figure}
\crefname{algocf}{fig.}{figs.}
\Crefname{algocf}{Figure}{Figures}

\section{Background on Module Theory}
\label{sec:prelim}

\input{prelim}

\section{\texorpdfstring{$\mathbb{Z}$}{Z}-weighted Automata}
\label{sec:Zlearning}
\input{Zlearning}

\section{P-finite Automata}
\input{holonomic_automata}

\subsection{Equivalence}
\label{subsec:Hequivalence}
\input{equivalence}

\subsection{Learning}
\label{subsec:Hlearning}
\input{learning}

 \newpage
 \begin{acks}
Mahsa Shirmohammadi is
supported by International Emerging Actions grant (IEA’22), by ANR grant VeSyAM (ANR-22-
CE48-0005) and by the grant CyphAI (ANR-CREST-JST).  
\end{acks}

\bibliographystyle{ACM-Reference-Format}
\bibliography{biblio}

\appendix

\input{app-same_rank}

\input{app-holonomic}

\end{document}

%% file: macros.tex
\usepackage{xcolor}
\usepackage{multirow}
\usepackage{tabularx}
\usepackage{tikz}
\usetikzlibrary{shapes.geometric}
\usetikzlibrary{automata,positioning}
\usetikzlibrary{fit}
\usetikzlibrary{matrix}
\usepackage{enumitem}

\usepackage{amsmath}
\usepackage{amsfonts}

\usepackage{todonotes}
\usepackage{xspace}

\usepackage{cleveref}
\usepackage{float}
\usepackage{thm-restate}

\usepackage[plain,linesnumbered]{algorithm2e}

\SetKwSwitch{Switch}{CaseO}{Other}{match}{with}{|\ }{|\ \_}{}{}
\SetKwSwitch{SwitchO}{Case}{OtherO}{match}{->}{|\ }{|\ \_}{}{}

\SetInd{0em}{1.3em}  

\AtEndPreamble{%
   \theoremstyle{acmdefinition}
   \newtheorem{remark}[theorem]{Remark}}

\newif\ifcomment
\commentfalse
\ifcomment
\newcommand{\cmahsa}[1]{\todo[color=green!40,inline]{{\bf Mahsa:} #1}}
\newcommand{\cben}[1]{\todo[color=blue!40,inline]{{\bf Ben:} #1}}
\newcommand{\calex}[2][]{\todo[color=teal!40,#1]{{\bf Alex:} #2}}
\newcommand{\cvincent}[2][]{\todo[color=red!40,#1]{{\bf Vincent:} #2}}

\else
\newcommand{\cmahsa}[2][]{}
\newcommand{\cben}[2][]{}
\newcommand{\calex}[2][]{}
\newcommand{\cvincent}[2][]{}

\fi

\newcommand{\NN}{\mathbb{N}}
\newcommand{\ZZ}{\mathbb{Z}}
\newcommand{\QQ}{\mathbb{Q}}

\newcommand{\A}{\mathcal{A}}   
\newcommand{\y}{\boldsymbol{y}} 
\newcommand{\init}{\boldsymbol{\alpha}}                      
\newcommand{\fin}{\boldsymbol{\beta}}                         
\newcommand{\sem}[1]{\llbracket #1 \rrbracket}  
\newcommand{\trans}{\mu}                        

\newcommand{\backF}[1]{B_{#1}}                  

\newcommand{\backS}[1]{\mathcal{B}_{#1}}        

\newcommand{\mspan}[2][]{\left\langle #2\right\rangle_{#1}}
\DeclareMathOperator{\rankop}{rank}
\newcommand{\rank}[1]{\rankop(#1)}

\newcommand{\PHankel}[2]{H(#1,#2)}
\newcommand{\AHankel}[3]{A_{#3}(#1,#2)}
\newcommand{\SeqC}{\mathcal{C}}                    
\newcommand{\SeqR}{\mathcal{R}}                    
\newcommand{\AH}{\mathcal{H}}                    

\newcommand{\colF}[2]{\mathrm{col}_{#1}(#2)}
\newcommand{\rowF}[2]{\mathrm{row}_{#1}(#2)}

\newcommand{\rowvec}[1]{\mathbf{#1}}

\newcommand{\forFB}[2]{F_{#1}^{#2}}               
\newcommand{\forSB}[2]{\mathcal{F}_{#1}^{#2}}     

\SetKwProg{Fn}{Function}{}{}
\SetKwFunction{Reset}{reset}%
\SetKwFunction{Update}{update}%
\SetKwFunction{Close}{close}%
\SetKwFunction{Close}{close}%
\SetKwFunction{LearnerKnows}{learner\_knows}%
\SetKwData{KwSome}{Some}
\SetKwData{KwNone}{None}
\SetKwInOut{Input}{Input}
\SetKwInOut{Output}{Output}
\SetKwRepeat{Do}{do}{while}
\SetKw{KwPick}{pick}

\SetKwFunction{FRecurs}{FnRecursive}%
\SetKwProg{Fn}{def}{\string =}{}

\SetAlgorithmName{Program}{Program}{Program}
\crefname{algocf}{prog.}{progs.}
\Crefname{algocf}{Program}{Programs}
\SetAlCapSkip{0.3cm}

\SetKwFunction{PLearner}{partial\_learner}%
\SetKwFunction{ELearner}{exact\_learner}%
\SetKwFunction{QEquiv}{$\QQ$\_equivalence\_oracle}%
\SetKwFunction{ZGenerators}{compute\_$\ZZ$\_generators}%
\SetKwFunction{ZCompute}{compute\_$\ZZ$\_automaton}%
\SetKwFunction{Membership}{membership\_oracle}%
\SetKwFunction{Equiv}{equivalence\_oracle}%
\SetKwFunction{HBuild}{build\_automata}%
\SetKwFunction{AddRow}{extend\_row}%
\SetKwFunction{AddCol}{extend\_col}%

\makeatletter
\renewcommand{\@algocf@capt@plain}{above}
\renewcommand{\algocf@caption@plain}{\box\algocf@capbox\vskip\AlCapSkip}%
\makeatother
\SetAlgoCaptionLayout{centerline}
\SetAlCapSkip{0.3cm}                  

\renewcommand{\vv}{\boldsymbol{v}}

\newcommand{\Z}{\ZZ}
\newcommand{\N}{\NN}
\newcommand{\Q}{\QQ}

\renewcommand{\vec}[1]{\boldsymbol{#1}}
\newcommand{\valpha}{\init}
\newcommand{\veta}{\fin}

%% file: abstract.tex
\begin{abstract}
  We introduce the class of P-finite automata.  These are a generalisation of weighted automata, in which the weights of transitions can depend polynomially on the length of the input word.  
  P-finite automata can also be viewed as simple tail-recursive programs in which the arguments of recursive calls can non-linearly refer to a variable that counts the number of recursive calls.
  The nomenclature is motivated by the fact that over a unary alphabet P-finite automata compute so-called P-finite sequences, that is, sequences that satisfy a linear recurrence with polynomial coefficients.  Our main result shows that P-finite automata can be learned in polynomial time in Angluin's MAT exact learning model.  This generalises the classical results that deterministic finite automata and weighted automata over a field are respectively polynomial-time learnable in the MAT model.
\end{abstract}

%% file: intro.tex
A central problem in computational learning  
is to determine a representation of a  function through information about its behaviour on specific inputs.  This problem encapsulates one of the main challenges in the analysis and verification of systems and protocols---namely, inferring an abstract model of a black-box system from a specification or a log of its behaviour.

In the case of functions represented by automata, one of most influential and well-known
formalisations of the learning problem is the \emph{minimally adequate teacher} (MAT) model, introduced by Dana Angluin~\cite{Angluin87}.
In this framework a learning problem is specified by a semantic class 
of functions and a syntactic class of representations (e.g., the class regular languages, represented by deterministic finite automata) and the goal of the learner is to output a representation of a given \emph{target function} by making \emph{membership} and \emph{equivalence} queries to a teacher.   In a membership query the algorithm asks the teacher the value of the target function on a specific argument, whereas in an equivalence query the algorithm asks whether its current hypothesis represents the target function and, if not, receives as counterexample an argument on which the hypothesis and target differ.   This  framework is sometimes referred to as \emph{active learning}, since the learner actively gathers information rather than 
passively receiving randomly chosen
examples, as in Valiant's PAC learning model.  Another difference with the PAC model is that in the latter the hypothesis output by the learner is only required to be approximately correct, while in the MAT model it should be an exact representation of the target function.

In the MAT model, we say that a given learning algorithm runs in polynomial time if its running time is polynomial in the shortest representation of the target concept and the 
length of the longest counterexample output by the teacher.  The running time is, by construction, an upper bound on the total number of membership and equivalence queries.
Among other contributions \cite{Angluin87} 
introduced the $L^*$ algorithm: a polynomial-time exact learning algorithm for regular languages, using the representation class of deterministic finite automata.  
The $L^*$ algorithm essentially tries to discover and distinguish the different Myhill-Nerode equivalence classes of the target language.  By now there are several highly optimized implementations of the basic algorithm, including in the LearnLib26 and Libalf packages~\cite{BolligKKLNP10,IsbernerHS15}.

For many applications, such as interface synthesis, network protocols, and compositional verification, deterministic finite-state automata
are too abstract and inexpressive to capture much of the relevant behaviour. This has motivated various
extensions of Angluin’s $L^*$ algorithm to more expressive models, such as non-deterministic, visibly pushdown, weighted, timed, register, and nominal automata~\cite{BolligHKL09,MichaliszynO22,HowarJV19,MoermanS0KS17}.   The current paper considers an extension of weighted automata.
The class of weighted automata over a field was introduced by Sch\"{u}tzenberger~\cite{Schutzenberger61b} and has since been widely studied in the context of probabilistic automata, ambiguity in non-deterministic automata, and formal power series.
A weighted automaton is a non-deterministic finite automaton whose transitions are decorated with constants from a weight semiring.  Here we focus on the case that the weight semiring is the field $\Q$ of rational numbers.
Although weighted automata over a field are strictly
more expressive and exponentially more succinct than deterministic automata, the class remains
learnable in polynomial time in the MAT model~\cite{exact-learning-wa}.  
By contrast, 
subject to standard cryptographic assumptions~\cite{AngluinK95}
non-deterministic finite automata are not learnable in the MAT model with polynomially many queries.

\textbf{Contributions of this paper.}
We introduce and study a generalisation of weighted automata, which we call \emph{P-finite automata}, in which each transition weight is a polynomial function of the length of the input word.  Over a unary alphabet, 
whereas weighted automata represent 
$C$-finite sequences (sequences that satisfy linear recurrences with constant coefficients),
P-finite automata represent so-called P-finite sequences (those that satisfy linear recurrences with polynomial coefficients).  P-finite sequences are  a classical object of study in combinatorics and the complexity analysis of algorithms~\cite{Tetrahedron}.  P-finite automata can thus be considered as a common generalisation of P-finite sequences and $\Q$-weighted automata. In Section~\ref{sec:overview} we also view weighted and P-finite automata as simple tail-recursive programs.

The main results of the paper involve
two different
developments of the problem of learning $\Q$-weighted automata, respectively involving more general and more specific representation classes.
\begin{itemize}
    \item
    The most important contribution concerns a generalisation of the algorithm of~\cite{exact-learning-wa} for learning $\Q$-weighted automata.
    We give a polynomial-time learning algorithm for the class of P-finite automata in the MAT model.  As a stepping stone to this result we show that the equivalence problem for P-finite automata is solvable in polynomial time.
       \item
       In a second direction 
    we consider the special case of the learning problem for
    $\Q$-weighted automata in which the target function is assumed to be integer valued. 
    Clearly the algorithm of~\cite{exact-learning-wa} can be applied in this case, but its final output and intermediate equivalence queries may be $\Q$-weighted automata.
    On the other hand, it was shown in~\cite{fliess1974matrices} that if a $\Q$-weighted 
    automaton gives an integer weight to every word then it has a minimal representation that is 
    a $\Z$-weighted automaton.
    Thus, in the case of an integer-valued target it is natural to ask for a learning algorithm that uses $\Z$-weighted automata as representation class.
    We give such an algorithm, running in polynomial time, and show how it can be implemented using division-free arithmetic.  The heart of this construction is to 
    give a polynomial-time procedure to decide whether a $\Q$-weighted automaton is $\Z$-valued and, if yes, to output an equivalent minimal $\Z$-weighted automaton.

\end{itemize}

\textbf{Related Work.}
In the case of a unary alphabet, P-finite automata are closely related to the matrix representations of P-finite sequences considered in~\cite{Reutenauer12}.  Over general alphabets P-finite automata can be seen as a very special case of the polynomial automata of~\cite{BenediktDSW17}.  However, while determining equivalence of P-finite automata is in polynomial time, checking equivalence of polynomial automata is non-primitive recursive.  The key difference is that in the case of P-finite automata one works with modules over univariate polynomial rings, which are principal ideal domains, rather than general polynomial rings, which are merely Noetherian.  The former setting yields much smaller bounds on the length of increasing chains of modules (compare, e.g., Proposition~\ref{prop:stablise} herein with~\cite[Theorem 2]{BenediktDSW17}).


The problems of learning automata with weights in principal ideal domains (such as the ring $\Z$ of integers and the ring $\Q[x]$ of univariate polynomials with rational coefficients) was investigated in~\cite{HeerdtKR020}.
That paper relies on the fact that finitely generated modules over principal ideal domains are Noetherian for 
the termination of the learning algorithm. 
The methods of the paper 
do not address the question of the query and computational complexity of the learning problem.  The paper also leaves open the question of learning minimal representations of a given target function.  Here we give a method that runs in polynomial time in the case of automata with weights in $\Z$ and $\Q[x]$ and that learns minimal representations.

%% file: overview.tex
\subsection*{Linear Tail-Recursive Programs.}
The weighted extensions of  automata that are currently studied in the literature are able to model  simple classes of tail-recursive programs, including linear recurrences. 
Consider \Cref{prog:mod2}, which reads a string of $a$'s letter-by-letter from the input, and  computes the  function~$f_1:\{a\}^* \to \ZZ$ such that 
\[f_1(a^k)=\begin{cases}
    2 & k \equiv 0 \pmod 2\\
    1 & \text{otherwise.}
\end{cases}\]


\SetKwFunction{Prog}{main}%
\SetKwFunction{ProgAux}{prog}%

\begin{minipage}{.60\linewidth}
\begin{algorithm}[H]
\Fn{\ProgAux{$y_1,y_2$}}{  
    \Switch{read()}{
        \lCase{None}{\KwRet{$y_1+y_2$}}
        \lCase{Some $a$}{\ProgAux{$2y_2,\frac{1}{2}y_1$}}
    }
}
\lFn{\Prog{}}{\ProgAux{$2,0$}}
\caption{A  linear tail-recursive program computing~$f_1$. }
\label{prog:mod2}
\end{algorithm}
\end{minipage}
\qquad
\begin{minipage}{.30\linewidth}
    \begin{tikzpicture}[shorten >=1pt,node distance=2cm,on grid,auto,accepting/.style=accepting by arrow] 
        \node[state, initial, initial text={$2$},accepting below,accepting text={$1$}] (q_1)  {$q_1$}; 
        \node[state,accepting below,accepting text={$1$}] (q_2) [right = of q_1] {$q_2$}; 
        \path[->] 
              (q_1) edge [bend left=30] node {$a:\frac{1}{2}$} (q_2)
          (q_2) edge [bend right=-30] node {$a:2$} (q_1);
    \end{tikzpicture}

\end{minipage}

\vskip0.7\baselineskip

The above program  can be modelled by a \emph{weighted automaton} with two states~$q_1$ and $q_2$, as depicted on the right  above. The states $q_i$ represent the output of~$f$, based on the congruence classes modulo~$2$. 
Intuitively speaking, the weight of a word is the sum of the weights of all runs  of the automaton over  the word, where the weight of a run is the product of weights of its starting state, of each transition taken along the run, and of its last state. For the automaton of~\Cref{prog:mod2} the non-zero initial weights are shown by incoming arrows to the states, whereas final weights are shown by outgoing arrows; each transition is also labelled by the letter~$a$ and its weights. 

Formally, a $\mathbb{Q}$-weighted automaton~$\A=(\init,\trans,\fin)$ of dimension $n$ over an alphabet $\Sigma$ is defined by 
the initial weight vector~$\init \in \QQ^{1\times n}$, a transition function $\trans: \Sigma \to \QQ^{n\times n}$ and the final weight vector~$\fin\in \QQ^{n\times 1}$. The  semantics of~$\A$, denoted by $\sem{\A}:\Sigma^* \to \mathbb{Q}$,  maps each word $w=\sigma_1\cdots \sigma_k$ to its weights computed as $\init \trans({\sigma_1}) \ldots \trans({\sigma_k}) \fin$. The automaton of~\Cref{prog:mod2} is formally defined as 
\begin{align*}
\init:=\begin{bmatrix}
2 & 0
\end{bmatrix} 
\qquad 
\trans(a):=\begin{bmatrix}
0 & \frac{1}{2}\\
2 & 0
\end{bmatrix} \qquad 
\fin:=\begin{bmatrix}
1\\
1
\end{bmatrix}\,.
\end{align*}

The automaton of~\Cref{prog:mod2} is a unary(-alphabet) automaton over $\mathbb{Q}$; it is well-known that unary weighted
automata over a field coincide with linear recurrence sequences~\cite{BerstelR10} over the field. Recall that a rational sequence~$\{u_i\}_{i=1}^\infty$ is a linear recurrence sequence of order $d$ if it satisfies a recurrence relation of the form 
\[u_n=c_{d-1}u_{n-1} + \ldots + c_1u_{n-d+1},\] 
where  $c_i\in \mathbb{Q}$ and $c_1\neq 0$. The Fibonacci sequence, for example, is given by $F_0 =F_1 = 1$ and $F_k = F_{k-1}+F_{k-2}$ for all $k\geq 2$. The corresponding Fibonacci automaton is defined by  
\begin{align*}
\init:=\begin{bmatrix}
1 & 1
\end{bmatrix} 
\qquad 
\mu(a):=\begin{bmatrix}
1 & 1\\
1 & 0
\end{bmatrix} \qquad 
\fin:=\begin{bmatrix}
1\\
0
\end{bmatrix} \,.
\end{align*} 
The automaton computes the $k$-th Fibonacci number as the weight of the input~$a^{k-1}$ through its semantics $\init \trans(a)^{k-1}\fin$.

In the general setting, a recursive program computing a  function $f:\Sigma^{*} \to \mathbb{Q}$ can be realised by a $\mathbb{Q}$-weighted automaton if its so-called Hankel matrix  has finite rank~\cite{berstel1988rational}. This characterization encompasses a rich class of linear tail-recursive programs,  where all assignments are  linear updates of the form $\y \leftarrow \y M $, where
$\y:=(y_1,\ldots,y_n)$ is a tuple of variables and $M\in \mathbb{Q}^{n\times n}$.
See~\Cref{prog:lrecusrive} for a schematic illustration of such linear recursive programs. 
An $\mathbb{Q}$-weighted automaton  for such programs is defined accordingly as $(\init,\trans,\fin)$ over the alphabet $\Sigma$.

\begin{algorithm}[t]
\Fn{\ProgAux{$\y$}}{  
    \Switch{read()}{
        \lCase{None}{\KwRet{$\y\fin$}}
        \lCase{Some $a$}{\ProgAux{$\y \trans(a) $}}
        \lCase{Some $b$}{\ProgAux{$\y \trans(b) $}}
        $\vdots$
    }
}
\lFn{\Prog{}}{\ProgAux{$\init$}}
\caption{Scheme of linear tail-recursive programs}
\label{prog:lrecusrive}
\vspace{-0.6cm}
\end{algorithm}

Before we proceed, we note that in weighted automata the weight growth of each word $w$   is bounded by $c^{|w|}$ for a fixed positive constant~$c\in \ZZ$. In the next section, we will see that, in our proposed extension of weighted automata, the weight growth of each word $w$  can be of magnitude $(c_1|w|)^{c_2|w|}$ where  $c_1,c_2\in \ZZ$ are fixed positive constants. 

\subsection*{Polynomial Tail-Recursive Programs.}

Our proposed P-recursive programs will have a program counter $x$, that initially is set to zero and monotonously increases by one after each input letter, in order to store the length of the word. The updates on each variable $y_i$ is now in the form $y_i \leftarrow \sum_{j=1}^nP_j(x)y_j$
where $P_1,\ldots, P_{n} \in \mathbb{Q}[x]$ are univariate polynomials with rational coefficients in indeterminate~$x$. 
\Cref{prog:fact}  computes the following function~$f_2:\{a,b\}^* \to \mathbb{Z}$ defined by 
\[f_2(w)=\begin{cases}
    (|w|+1)! & \text{ if $w$ contains an odd number of $b$'s,}\\
    0 & \text{otherwise.}
\end{cases}\]


\SetKwFunction{Prog}{main}%
\SetKwFunction{ProgAux}{prog}%

\begin{minipage}{.52\linewidth}
\begin{algorithm}[H]
\Fn{\ProgAux{$y_1,y_2,x$}}{  
    \Switch{read()}{
        \lCase{None}{\KwRet{$xy_2$}}
        \lCase{Some $a$}{\ProgAux{$x y_1,x y_2,x+1$}}
        \lCase{Some $b$}{\ProgAux{$x y_2,x y_1,x+1$}}
    }
}
\lFn{\Prog{}}{\ProgAux{$1,0,1$}}
\caption{A  P-recursive program computing~$f_2$. }
\label{prog:fact}
\end{algorithm}
\end{minipage}
\qquad\qquad
\begin{minipage}{.35\linewidth}
    \begin{tikzpicture}[shorten >=1pt,node distance=2cm,on grid,auto,accepting/.style=accepting by arrow] 
        \node[state, initial, initial text={$1$},accepting below,accepting text={$0$}] (q_1)  {$q_1$}; 
        \node[state,accepting,accepting below,accepting text={$x$}] (q_2) [right = of q_1] {$q_2$}; 
        \path[->] 
            (q_1) edge [loop above]  node [align=center]{$a:x$\\$b:0$} ()
            edge [bend left=30] node [align=center]{$a:0$\\$b:x$} (q_2)
            (q_2) edge [bend right=-30] node [align=center]{$a:0$\\$b:x$} (q_1)
            edge [loop above] node [align=center]{$a:x$\\$b:0$} ();
    \end{tikzpicture}
\end{minipage}

\vskip0.8\baselineskip

We show that such $P$-recursive programs can be realised by our proposed extension of weighted automata, which we call \emph{P-finite automata}. This extension can be thought of as a symbolic weighted automata where transition  weights, as well as final weights, are parameterized by an indeterminate~$x$. Along the execution of a P-finite automaton over an input word, the value of indeterminate~$x$ stores the length of the input read so far.  

In the P-finite automaton representing~\Cref{prog:fact} there are two states corresponding to the variables~$y_1$ and $y_2$.
As is the case for weighted automata, the weight of a word is the sum of the weight of all runs of the automaton over  the word, where the weight of a run is the product of  weights of its starting state, of each transition taken along the run, and of its last state. The main difference is that the transition and final weights change in every step, as the value of $x$ gets updated after every new input letter. 
For instance, the weight of~$ab$ is $3!$ computed by 
 \[\underbrace{1}_{\text{initial weight of } q_1}\cdot \overbrace{1}^{\text{weight of } q_1 \xrightarrow{a : \, 1 } q_1}\cdot \overbrace{2}^{\text{weight of } q_1 \xrightarrow{b:\, 2 } q_2}\cdot \underbrace{3}_{\text{final weight of } q_2: \, 3}\]

Formally, a P-finite automaton~$\A=(\init,\trans,\fin(x))$ of dimension $n$ over an alphabet $\Sigma$ is defined by 
the initial weight vector~$\init \in \QQ^n$, a transition function $\trans: \Sigma \to \QQ[x]^{n\times n}$, and the final weight vector~$\fin(x)\in \QQ[x]^n$. 
In the sequel, for simplicity  of  notations we use $\mu(\sigma,k)$ instead of $\mu(\sigma)(k)$, with $\sigma \in \Sigma$ and $k\in \NN$. 
The  semantics of~$\A$, denoted by $\sem{\A}:\Sigma^* \to \mathbb{Q}$,  maps each word $w=\sigma_1\cdots \sigma_k$ to  \[\sem{\A}(w):=\init \trans({\sigma_1},1) \ldots \trans({\sigma_k},k) \fin(k+1)\, .\] 
Although the initial vector $\init$ is a vector of rationals, one can also look at it as a vector of polynomials (similar to the final vector $\fin(x)$) that is always evaluated on $0$ as it would lead to an equivalent semantics $\init(0) \trans({\sigma_1},1) \ldots \trans({\sigma_k},k) \fin(k+1)$. The automaton of~\Cref{prog:fact} is formally defined as 
\begin{align*}
\init:=\begin{bmatrix}
1 & 0
\end{bmatrix} 
\qquad 
\trans(a):=\begin{bmatrix}
x & 0\\
0 & x
\end{bmatrix} \qquad 
\trans(b):=\begin{bmatrix}
0 & x\\
x & 0
\end{bmatrix} \qquad 
\fin:=\begin{bmatrix}
0\\
x
\end{bmatrix}\,.
\end{align*} 

Unary   P-finite
automata coincide with  monic $P$-recursive sequences. A rational sequence~$\{u_i\}_{i=1}^\infty$ is a (monic) $P$-recursive sequence of order $d$ if it satisfies a recurrence relation of the form 
\[u_n=P_{d-1}u_{n-1}+\ldots P_1u_{n-d+1},\] 
where  $P_i\in \mathbb{Q}[x]$ and $P_1\neq 0$. 
Another example of monic $P$-recursive sequences comes from the famous  recurrence for the number of involutions, found by Heinrich August Rothe in 1800.
An involution on  a set $\{1, 2, \ldots ,k\}$ is a self-inverse permutation. 
The number of involutions, including the identity involution,   is given by $I_0=I_1=1$ and 
$I_{k}=I_{k-1}+(k-1)I_{k-2}$
for $k\geq 2$. The corresponding P-finite automaton is defined by  
\[\init=\begin{bmatrix}
    1 & 1
    \end{bmatrix} \qquad \mu(a)=\begin{bmatrix}
    1 & 1 \\
    x & 0
    \end{bmatrix} \qquad \fin=\begin{bmatrix}
    1 \\
    0
    \end{bmatrix}\,.\]
The P-finite automaton computes the number of involutions of $\{1,\ldots,k\}$ as the weight of the input~$a^{k-1}$ through its semantics $\init \prod_{i=1}^{k-1}\mu(a,i)\fin(k)$.
See \Cref{prog:precusrive} for a schematic illustration of a class of  polynomial tail-recursive programs that can be realized by a P-finite automata.     

\begin{algorithm}[t]
\Fn{\ProgAux{$\y,x$}}{  
    \Switch{read()}{
        \lCase{None}{\KwRet{$\y\fin(x)$}}
        \lCase{Some $a$}{\ProgAux{$\y M_a(x),x+1 $}}
        \lCase{Some $b$}{\ProgAux{$\y M_b(x),x+1 $}}
    }
}
\lFn{\Prog{}}{\ProgAux{$\init,1$}}
\vspace{-0.3cm}
\caption{Scheme of P-recursive programs}
\label{prog:precusrive}
\end{algorithm}







\subsection*{P-Solvable Loops and Extensions}

The model of P-recursive programs (or P-finite automata) bears similarities with the notion of $P$-solvable loops~\cite{Kovacs08} and its extensions~\cite{humenberger2017, humenberger2017invariant}. 
The latter are studied in the context of program analysis and invariant synthesis in particular.

The class of P-solvable loops is subsumed by that of \emph{linear} tail-recursive programs, as P-solvable loops allow only linear updates of program
variables~\cite{Kovacs08}.  We have also the class of \emph{extended P-solvable loops}~\cite{humenberger2017,humenberger2017invariant}, in which the sequence of values assumed by a program variable is
a sum of \emph{hypergeometric sequences}. A hypergeometric sequence $(u_n)_{n=0}^\infty$ is one that satisfies a polynomial recurrence $u_n = r(n)u_{n-1}$ for all $n\geq 1$, 
where $r(x)\in \Q(x)$ is a rational function. 
The class of extended P-solvable loops is thus incomparable with P-finite automata.
On the one hand, hypergeometric recurrences allow multiplication by rational functions (such as $r(x)$ above), not just polynomials.  On the other hand P-finite automata over a unary alphabet can define sequences that are not sums of hypergeometric sequences (see~\cite[Section 10]{Reutenauer12}).


\SetKwFunction{Prog}{main}%
\SetKwFunction{ProgAux}{prog}%

\begin{minipage}{.35\linewidth}
\begin{algorithm}[H]
\While{true}{
    $a$ := $2 (x+1) (x+\frac{3}{2})  a $\;
    $b$ := $4  (x+1)  b$\;
    $c$ := $\frac{1}{2} (x+\frac{3}{2}) c$\;
    $x$ := $x+1$
}
\end{algorithm}
\end{minipage}
\qquad
\begin{minipage}{.6\linewidth}
  \begin{algorithm}[H]
\Fn{\ProgAux{$a,b,c,x$}}{  
    \Switch{read()}{
        \lCase{None}{\KwRet{$(a,b,c)$}}
        \lCase{Some \_}{\ProgAux{$2 (x+1)(x+\frac{3}{2})a,4(x+1)b,\frac{1}{2}(x+\frac{3}{2})c ,x+1$}}
    }
}
\end{algorithm}
\end{minipage}

\vskip0.8\baselineskip

The program shown above on the left is an example of an extended P-solvable loop, taken from~\cite{humenberger2017invariant}. The corresponding P-recursive program is shown on the right. Since the focus of~\cite{humenberger2017invariant} is on invariant generation they consider loops that run forever.  In our P-recursive programs any input letter invokes the recursive call.

\subsection*{Learning Algorithm}

The high-level structure of the algorithm for learning $P$-finite automata is shown in the diagram below.  The algorithm consists of a main procedure \texttt{exact\_learner} and a subroutine \texttt{partial\_learner}.  
It is also not assumed to know \emph{a priori} an upper bound $n$ on the number of states of the target automaton nor a degree bound $d$ on the polynomials appearing therein.  Hence the procedure 
\texttt{exact\_learner} searches through pairs of possible values of $d$ and $n$ and for each such pair it calls 
\texttt{partial\_learner} that tries to learn a target automaton subject to these bounds.

\begin{center}
    \scalebox{0.8}{
    \begin{tikzpicture}[
        state/.style={draw,rounded corners,align=center}
        ]
        \node[state] (Vnd) {\begin{tabular}{c}update\\$d,n$\end{tabular}};
        \node[state,right=1cm of Vnd] (Bound) {\begin{tabular}{c}compute\\timeout $\ell$\end{tabular}};
        \node[state,right=0.4cm of Bound] (Hyp) {\begin{tabular}{c}construct\\ hypothesis automaton\end{tabular}};
        \node[state,below right=0.7cm and 2cm of Bound] (Add) {add one row and at most one column};
        \node[state,right=5.7cm of Bound] (OK) {\begin{tabular}{c}return learned\\automaton\end{tabular}};

        \node[draw,dashed,minimum width=8.5cm,minimum height=2.7cm,above right=-2.5cm and 0.2cm of Bound] (Exact) {};
        \node[above left=0cm and -3cm of Exact] (ExactText) {\texttt{partial\_learner}};

        \node[draw,minimum width=14cm,minimum height=3.5cm,above right=-2.7cm and -5cm of Bound] (EExact) {};
        \node[above left=0cm and -2.6cm of EExact] (EExactText) {\texttt{exact\_learner}};
        
        \draw[->] (Vnd) edge node[auto] {} (Bound);
        \draw[->] (Bound) edge node[auto] {} (Hyp);
        \draw[->] (Hyp) edge node[above] {Teacher:} node[below] {OK} (OK);
        \draw[->] ([xshift=1.6cm]Hyp.south) to node[auto,swap,yshift=0.3cm,xshift=0.08cm] {Teacher: counterexample} ([xshift=-0.8cm]Add.north);
        \draw[->] ([xshift=-0.6cm]Add.north) to  node[swap,auto,yshift=-0.2cm] {if $\#\text{columns} < \ell$} ([xshift=1.8cm]Hyp.south);
        \draw[->] (Add) to [in=-35,out=180] node[auto,pos=0.8] {if $\#\text{columns}\geq \ell$} (Vnd);
    \end{tikzpicture}
    }
\end{center}

The subroutine \texttt{partial\_learner} can be seen as a
generalisation of the algorithm of~\cite{exact-learning-wa} for learning $\Q$-weighted
automata.  As in~\cite{exact-learning-wa}, the basic data structure, which we call the \emph{table}, is a finite fragment of the Hankel matrix of the
target function $f:\Sigma^*\rightarrow \Q$.
Formally speaking, the table is a finite matrix whose rows and columns are labelled by words and such that the entry with index $(u,v)\in
\Sigma^* \times \Sigma^*$ is $f(uv)$.  We will later on denote this table by $\PHankel{\SeqR}{\SeqC}$ where $\SeqR$ and $\SeqC$ are the sequences of words labelling the rows and columns of the table. The table is used to construct a
hypothesis automaton.  This involves making membership queries to interpolate polynomials
that label the state-to-state transitions of the automaton.
Since there is a bound $d$ on the maximum degree of the polynomials,
 the process of interpolation is reduced to solving a system of linear equations.

Once constructed, the hypothesis automaton is passed to the
teacher.  If the hypothesis is correct, the algorithm terminates and returns the hypothesis automaton; if it is incorrect, the
counterexample given by the teacher is used to augment the table by
adding a new row and at most one column (using membership queries to fill in the missing table entries).  After augmenting the table, a new hypothesis automaton can be constructed.

For any given run of \texttt{partial\_learner},
since the degree bound $d$ may not be sufficient to learn the target automaton, there is a timeout $\ell$ (a function of $d$ and $n$) on the run of 
\texttt{partial\_learner}.  If the timeout is reached then the run is abandoned and control returns to \texttt{exact\_learner}.

Going back to the case of $\QQ$-weighted automata,
the termination (and polynomial-time bound) of the learning algorithm of~\cite{exact-learning-wa} relies on the classical result
of Carlyle and Paz that a function $f:\Sigma^*\rightarrow \Q$ is recognisable
by a $\Q$-weighted automaton if and only if its Hankel matrix has finite rank.
The idea is that every unsuccessful equivalence query results in 
the rank of the table increasing by one, and so the number of equivalence queries is at most the rank of the Hankel matrix of the target function.  Such a result is not available in
the case of P-finite automata.

The termination proof and polynomial complexity bound for \texttt{exact\_learner} rely on an analysis of the timeout.  For this we associate with a run
of \texttt{partial\_learner} 
an increasing chain of submodules over the
polynomial ring $\Q[x]$, whose length is the
number of equivalence queries.  Since $\Q[x]$ is a Noetherian ring,
such a chain must have finite length.  We give a novel fine-grained analysis of the maximum length of 
 an increasing chain of modules over $\Q[x]$ to guarantee that we will learn the target automaton within the timeout if the degree bound parameter is sufficiently large.   (This analysis even allows us to obtain a polynomial bound on the overall computational and query complexity of our learning algorithm.) We highlight that in contrast to the case of $\Q$-weighted automata, the length of this chain of modules depends on the length of the counterexamples returned by the teacher.  As an intermediate result, we use this analysis to show that
equivalence of P-finite automata is decidable in polynomial time.

Our analysis of increasing chains of modules over $\Q[x]$ applies
equally well to $\Z$.  We use the version for $\Z$ to give a polynomial-time
algorithm to decide whether or not the function recognised by a given
$\Q$-weighted automaton is $\Z$-valued.  We then observe that such an
algorithm can be used to reduce the problem of learning $\Z$-weighted
automata to that of learning $\Q$-weighted automata.

%% file: prelim.tex
Let $R$ be a commutative ring with unity. An \emph{ideal} of $R$ is an additive subgroup $I\subseteq R$ such that $ra \in I$ for all $r\in R$ and $a\in I$.   The ring~$R$ is said to be a \emph{principal ideal domain} (PID) if every ideal~$I$ is generated by a single element, that is, there is some $a\in R$ such that $I= \{ ra : r \in R \}$.  We will mainly work with $\Z$ and $\Q[x]$, which  are both PIDs.
 
An $R$-module $M$ is an abelian group together with a scalar multiplication $(\cdot): R \times M \rightarrow M$ such that,
\begin{itemize}
    \item $r\cdot(m_1+m_2)=rm_1+rm_2$,
    \item $(r_1+r_2)\cdot m=r_1\cdot m+r_2 \cdot m$,
    \item $(r_1r_2)\cdot m=r_1\cdot (r_2 \cdot m)$, $1_R\cdot m=m$,
\end{itemize}
for all scalars $r, r_1, r_2 \in R$, and for all elements $m, m_1, m_2 \in M$.   A key example of an $R$-module is $R^n$, where $n\in \N$, in which addition and scalar multiplication act pointwise.

Let $M$ be an $R$-module.  A \emph{submodule} of $M$ is a subgroup that is closed under scalar multiplication.  A subset~$\{\vv_i : i\in I\} \subseteq M$ is said to be \emph{linearly independent} if an $R$-linear combination 
$\sum_{i\in I} r_i\vv_i$ is only zero if all the $r_i$ are zero.
We write $\mspan[R]{\vv_i : i\in I}$ for the $R$-\emph{span} of  the $\vv_i$, defined by
\[
\mspan[R]{\vv_i : i\in I} := \left\{\sum_{i\in I}r_i \cdot \vv_i : r_i \in R \right\}.
\]
We say that $\{\vv_i : i\in I\}$ \emph{generates} $M$ if $M = \mspan[R]{\vv_i : i\in I}$. If the $\vv_i$ are, in addition, linearly independent, then we say that $\{\vv_i : i\in I\}$ is a \emph{basis} of $M$.   If $R$ is a PID then all submodules $M$ of $R^n$ have a basis and all bases have the same cardinality. 

A key difference between modules and vector spaces is that one can have proper inclusions between modules of the same rank.  For example, we have that $15\Z \subsetneq 3\Z \subsetneq \Z$ are all rank-1 submodules of $\Z$.  However it remains true that all finitely generated $R$-modules are \emph{Noetherian}: every strictly increasing chain of submodules of $M$ is  finite.   A crucial ingredient in the analysis of our algorithms is an upper bound on the length of strictly increasing chains of modules in~$R^n$. For this, we use the Smith Normal Form. 

Let $R$ be either $\Z$ or $\Q[x]$ and let $M = \mspan[R]{\vv_1 , \ldots,\vv_m}$ be a finitely generated $R$-module of rank $r$.
   Using the Smith Normal Form~\cite{SNF}, one can show that from the set~$\{\vv_1 , \ldots,\vv_m\}$ of generators of $M$, we can compute, in polynomial time, a $R$-basis $\boldsymbol{f}_1,\ldots,\boldsymbol{f}_n$ of $R^n$ and elements~$d_1,\ldots,d_r \in R$ such that~$d_1\boldsymbol{f}_1,\ldots,d_r\boldsymbol{f}_r$ is an $R$-basis of $M$. 
 To be specific, 
 the matrix 
 \[A := \begin{bmatrix} \vv_1 & \ldots & \vv_m \end{bmatrix}\] 
 can be put in Smith Normal Form, that is, 
 $A$ can be written as $S \tilde{A} T$ where
 $S \in R^{n\times n}$ and $T \in R^{m \times m}$ are invertible matrices, and   \[\tilde{A}=diag(d_1,\ldots,d_r,0,\ldots,0)\] is  a diagonal matrix 
 such that $d_i\,|\, d_{i+1}$  for all $1 \leq i < r$.
 Define $D_i(A)$ to be the greatest common divisor of the $i \times i$ minors of $A$ for $i=0,\ldots,r$,
(so that $D_0(A) = 1$). It is known~\cite{newman1997smith} that for all $1\leq i \leq r$ we have
 \begin{align}
 \label{eq:diSNM}
     d_i = \frac{D_i(A)}{D_{i-1}(A)} \, .
 \end{align} In the above-mentioned decomposition of $A$ as $S \tilde{A} T$, the columns of $S$ are in fact the $R$-basis $\boldsymbol{f}_1,\ldots,\boldsymbol{f}_n$ of $R^n$, and $S \tilde{A}$ is a matrix with columns $d_1\boldsymbol{f}_1,\ldots,d_r\boldsymbol{f}_r$, representing an $R$-basis of $M$.

\begin{proposition}
Let $n,k\in \NN$ and $M_0\subsetneq M_1 \subsetneq\cdots \subsetneq M_k$ be a strictly
increasing chain of submodules of~$\mathbb{Z}^n$, all having the
same rank $r \leq n$.
Assume that $M_0$ is generated by a collection of vectors
whose entries have absolute value at most $B$.  Then
$k\leq r \log B+ \frac{r}{2} \log r$.
\label{prop:stablise}
\end{proposition}
  
\begin{proof}[Proof Sketch]
By assumption, there are vectors $\vv_1,\ldots,\vv_m \in \mathbb{Z}^n$ that generate $M_0$ and
whose entries have absolute value at most $B$. 
Using Smith Normal Form, there exists a basis $\boldsymbol{f}_1,\ldots,\boldsymbol{f}_n$ of $\mathbb{Z}^n$
and positive integers $d_1,\ldots,d_r$, such that 
$d_1\boldsymbol{f}_1,\ldots,d_r\boldsymbol{f}_r$ is a basis of $M_0$. Furthermore, by Equation~\eqref{eq:diSNM} it follows  that
$d_1\cdots d_r$ is the greatest common divisor of all $r\times r$
minors of the $n\times m$ matrix with columns
$\vv_1,\ldots,\vv_m$.  By Hadamard's inequality it follows
that $d_1\cdots d_r\leq B^rr^{r/2}$.

Let $M$ be the module generated by $\boldsymbol{f}_1,\ldots,\boldsymbol{f}_r$.
Since the modules $M_0,\ldots,M_k$ all have rank $r$, they are all 
contained in $M$. 
Recall that the index~$[M:M_0]$ of a subgroup $M_0$ in the group $M$,   is the number of  cosets of $M_0$ in $M$. 
Observe that the index
$[M:M_0]$  is $d_1\cdots d_r$.  We
also have $[M_{k+1}:M_k] \geq 2$ for $k=0,\ldots,n-1$ since
$M_k$ is a proper submodule of $M_{k+1}$.  It follows that
$n \leq \log(d_1 \cdots d_r) \leq r \log B+ \frac{r}{2} \log r$.
\end{proof}

\begin{remark}
\label{remark:stabilise}
The bound in \Cref{prop:stablise} is tight: consider $\Vec{e_1}=\begin{bmatrix}
    1 & 0 & \cdots & 0
\end{bmatrix}^\top\in\ZZ^n$ and some positive integer $b$. The strictly increasing chain of modules of rank $1$:
\[\mspan[\ZZ]{2^b \Vec{e_1}} \subsetneq \mspan[\ZZ]{2^b\Vec{e_1},2^{b-1}\Vec{e_1}} \subsetneq \cdots \subsetneq \mspan[\ZZ]{2^b\Vec{e_1},2^{b-1}\Vec{e_1},\cdots,\Vec{e_1}}\]
has length $b$, which is the bound given by \Cref{prop:stablise}.
\end{remark}

 In the next Proposition, we generalize \Cref{prop:stablise}   to the PIDs for which there exists a well-defined greatest common divisor function. A detailed proof can be found in \Cref{sec:proof_seq_length}.

\begin{restatable}{proposition}{lemmagenerallength}
\label{lemma:general_length}
Let $n,k \in \NN$. 
  Let $R$ be a PID and 
  $M=\langle \vv_1,\ldots, \vv_m \rangle$ be a $R$-submodule of $R^n$. Let~$A$ be the $n\times m$ matrix whose $i$-th column is $\vv_i$.   Let $
  M\subsetneq M_1\subsetneq M_2\subsetneq \ldots \subsetneq M_k
  $ be a strictly increasing chain of $R$-submodules of $R^n$, all having the same rank $r\leq n$. Then $k$ is bounded by the number of (not necessarily distinct) prime factors of $D_r(A)$.
\end{restatable}

%% file: Zlearning.tex
In this section, we start by giving a procedure to decide  in polynomial
time whether a $\QQ$-weighted automaton computes an integer-valued function. 
For every ``yes'' instance our procedure returns an
equivalent $\ZZ$-weighted automaton and for every ``no'' instance it
returns a word whose weight is non-integer.  This algorithm can be
regarded as an effective (and computationally efficient) version of
the well-known fact that $\QQ$ is a Fatou extension of $\ZZ$~\cite[Chapter 7]{BerstelR10}.  
As a corollary of the above procedure, we give a polynomial-time reduction of the exact learning
problem for $\ZZ$-automata to the exact learning problem for
$\QQ$-automata. (One can similarly reduce the exact learning problem for automata with weights in the ring $\QQ[x]$ to that for automata with weights in the quotient field $\QQ(x)$.)


\subsection{\texorpdfstring{$\ZZ$}{Z}-valuedness of \texorpdfstring{$\QQ$}{Q}-automata}


Let $\A=(\init,\trans,\fin)$ be a \emph{$\QQ$-weighted
  automaton} of dimension $n$ over alphabet $\Sigma$. Here
$\init \in \QQ^{1\times n}$,
$\trans(\sigma) \in \QQ^{n\times n}$ for all $\sigma\in\Sigma$,
and $\fin \in \QQ^{n\times 1}$.  We say that such an automaton
$\A$ is 
\emph{$\mathbb{Z}$-weighted} if all entries of $\init,\fin$ and those of the
matrices $\trans(\sigma)$ are integers.  Let $I_n$ be the $n\times n$ identity matrix. We extend
$\trans$ to a map $\trans:\Sigma^*\rightarrow \QQ^{n\times n}$ by
writing $\trans(\varepsilon) := I_n$  and
$\trans(w \sigma) := \trans(w)\mu(\sigma)$ for all
$\sigma\in \Sigma$ and $w\in\Sigma^*$.  The semantics of $\A$, that is, the function computed by $\A$, is given by $\sem{\A} : \Sigma^* \rightarrow \mathbb{Q}$ with $\sem{\mathcal{A}}(w) := \init \trans(w) \fin$. Automata $\mathcal{A}_1,\mathcal{A}_2$ over the same alphabet $\Sigma$ are said to be \emph{equivalent} if $\sem{\mathcal{A}_1} = \sem{\mathcal{A}_2}$. An automaton $\mathcal{A}$ is \emph{minimal} if there is no equivalent automaton with fewer states.

Define the \emph{forward reachability set} of $\A$ to be
$\{ \init \trans(w) : w\in \Sigma^*\}$ and define the \emph{backward
  reachability set} to be $\{ \trans(w)\fin : w\in\Sigma^*\}$.  The
\emph{forward space} and \emph{forward module} of $\A$ are
respectively the $\QQ$-subspace of $\QQ^n$ and
$\ZZ$-submodule of $\QQ^n$ spanned by the forward reachability
set, \emph{viz.},
\[\mspan[\QQ]{\init \trans(w) : w\in \Sigma^*} \qquad \text{ and } \qquad \mspan[\ZZ]{\init \trans(w) : w\in \Sigma^*}\,.\] 

The backward space and backward module are defined analogously.
The forward space is the smallest (with respect to inclusion) vector space that contains $\init$ and is closed under post-multiplication by $\trans(\sigma)$.  The forward module is likewise the smallest module that contains $\init$ and is closed under post-multiplication
by $\trans(\sigma)$.  Analogous statements apply to the backward space and backward module.

Let $F \in \QQ^{m_f \times n}$ with $m_f\leq n$ be
a matrix whose rows form a basis of the forward space of~$\A$. It is known that there are unique $\init_f \in \Q^{1\times m_f}$, $\fin_f\in\Q^{m_f \times 1}$ and 
$\trans_f(\sigma)\in \Q^{m_f\times m_f}$, for all $\sigma \in \Sigma$, such that:
\begin{gather}\init_fF=\init \qquad \trans_f(\sigma)F=F\trans(\sigma) \qquad \fin_f=F\fin \,.
\label{eq:forward}
\end{gather}
Similarly, let
$B \in \QQ^{n \times m_b}$ with $m_b\leq n$ be a matrix whose columns form a basis of the
backward space of~$\A$. 
It is known that there are unique $\init_b \in \QQ^{1\times m_b},\fin_b\in \QQ^{m_b\times 1}$ and 
$\trans_b(\sigma)\in \Q^{m_b\times m_b}$, for all $\sigma \in \Sigma$, such that:
\[\init_b=\init B \qquad B\trans_b(\sigma)=\trans(\sigma)B \qquad B\fin_b=\fin \,.\] 
The automaton~$\A_f=(\alpha_f,\mu_f,\beta_f)$ is  a \emph{forward conjugate} of~$\A$, and the automaton $\A_b=(\alpha_b,\mu_b,\beta_b)$ is a \emph{backward conjugate} of $\A$. These automata are equivalent to~$\A$, meaning that  \[\sem{\A_f}=\sem{\A_b}=\sem{\A}\,. \]

\bigskip

The procedure to decide $\ZZ$-valuedness of  $\QQ$-automata is  a variant of the classical minimisation algorithm for $\QQ$-weighted
automata and it is described in~\Cref{alg:Z_compute}. Below, we first work through a subroutine used in the algorithm.

It is classical that given an automaton $\A$ we can compute
in polynomial time a $\QQ$-basis of the forward vector space that is comprised of
vectors in the forward reachability set.  
An analogous result holds
for the backward space~\cite{doi:10.1137/0221017,kiefer2020notes}.  
The forward module need not be finitely
generated in general, but  it will be finitely generated if the forward
reachability set is contained in $\ZZ^n$.   

The procedure \ZGenerators, shown in \Cref{alg:forward module}, is a
polynomial-time algorithm that, for an input~$\QQ$-automaton, 
either outputs a finite basis of the
forward module of $\A$ or a non-integer vector in the forward
reachability set.  
Intuitively, it builds a set of words $W$, starting from $\{ \varepsilon\}$, by adding words that augment the module  $\mspan[\ZZ]{\init\trans(u) : u \in W}$. 
When no such word can be found, the set $\{\init\trans(u) : u \in W\}$ will form a generating set for the forward module. 

Notice that the procedure is based on a two-pass search:
first we search for words that increase the rank of the forward module and then for words that augment the forward module while the rank is stable.  This allows us to obtain a polynomial-time running bound through a single application of \Cref{prop:stablise} to the second phase of the search.  
We do not know if it is possible 
to obtain a polynomial bound under arbitrary search orders. 


\begin{algorithm}[t]
\Fn{\ZGenerators{$\A$}}{
  $W := \{\varepsilon\}$\;
  \tcp{\textcolor{blue}{Finding words that increase the rank}}
  \SetAlgoVlined
  \SetAlgoShortEnd
  
  \While{there is $(w,\sigma) \in W\times \Sigma$ such that $\init\trans(w\sigma) \not\in \mspan[\QQ]{\init\trans(u) : u \in W}$}{\label{linealgo:first-while}
  
    $W := W \cup \{w\sigma\}$\;
    \lIf{$\init\trans(w\sigma) \not\in \ZZ^n$}{\KwRet{$w\sigma$}}
  }
  \tcp{\textcolor{blue}{Finding words that augment the module}}
  \While{there is $(w,\sigma) \in W\times \Sigma$ such that $\init\trans(w\sigma) \not\in \mspan[\ZZ]{\init\trans(u) : u \in W}$}{\label{linealgo:second-while}
    $W := W \cup \{w\sigma\}$\;
    \lIf{$\init\trans(w\sigma) \not\in \ZZ^n$}{\KwRet{$w\sigma$}}
  }
  \KwRet{$W$}
}
\caption{Computing generators of the forward module or a counterexample.}
\label{alg:forward module}
\end{algorithm}


\begin{proposition}
The procedure \ZGenerators, described in \Cref{alg:forward module}, is a polynomial-time algorithm that given a 
$\mathbb{Q}$-automaton $\mathcal{A}=(\init,\trans,\fin)$ of
  dimension $n$ over alphabet $\Sigma$ either outputs a finite set of words $W$ generating the forward module of $\A$, namely,  \[\mspan[\ZZ]{\init\trans(w) : w \in W}\, =\mspan[\ZZ]{\init\trans(w) : w \in \Sigma^*} ,\]
  or else a word
  $w\in\Sigma^*$ such that $\init \trans(w)\not\in \ZZ^n$.
  \label{prop:promise}
\end{proposition}

\begin{proof}
Write the entries of $\init,\fin$ and $\trans(\sigma)$,
with $\sigma\in\Sigma$, as fractions over a common denominator and let $B$ be an upper bound of the numerators and denominator of the
resulting fractions.  Note that the bit size of $B$ is
polynomially bounded in the length of the encoding of
$\mathcal{A}$.

The first {\bf while}-loop, in Line \ref{linealgo:first-while}
computes a set of words 
$W_0 \subseteq \Sigma^{\leq n}$ such that $\{ \init\trans(w) : w \in W_0\}$ is a $\QQ$-basis of the forward space of $\A$.
By construction, the dimension of the space spanned 
by the set $\{ \init\trans(w) : w \in W_0\}\leq |W_0|\leq n$, which shows that the first {\bf while}-loop terminates after at most $n$ iterations.

Below, we prove that the second {\bf while}-loop, in Line \ref{linealgo:second-while},
 terminates in polynomial
time in  the length of the encoding of
$\mathcal{A}$.  Let
$W_0,W_1,W_2,\ldots$ be the successive values of the variable $W$ during the second loop.
For all $k\in \NN$, let $M_k$ be the $\ZZ$-module $\mspan[\ZZ]{\init\trans(w) : w \in W_k}$. Then
$M_0 \subsetneq M_1 \subsetneq \cdots$ is a strictly increasing
sequence of $\ZZ$-modules, all having the same rank (namely the size of~$W_0$, that is the
dimension of the forward space).

Recall that the length of words in $W_0$ is at most $n$.
A simple induction on the length of  words allows us to show that for all $w \in \Sigma^*$, the entries of $\init\trans(w)$ have numerators and denominators bounded by $n^{|w|-1}B^{|w|}$. In particular, we obtain that the entries in $\{\valpha\mu(w):w\in W_0\}$ are bounded by $n^{n-1}B^n$.

Let
$k_0:= n(n-\frac{1}{2})\log n + n^2 \log B$.  Suppose that all modules $M_0,M_1, \cdots $ contain only integer vectors.  Then
Proposition~\ref{prop:stablise} shows that the above sequence modules
has length at most $k_0$.  The only other possibility is that for
some $k\leq k_0$ we have $M_k \not\subseteq \ZZ^n$ and hence
$\init \mu(w) \not\in \ZZ^n$ for some word~$w\in W_k$.  In either case,
the number of iterations of the while loop is at most $k_0$.

It follows that each set $W_k$ consists of at most $k_0+n$ words, each
of length at most $k_0+n$.  Thus the set of vectors
$\{ \init \trans(w) : w\in W_k\}$ has description length polynomial in
$\mathcal{A}$.  Each iteration of the while loop involves solving
$|W|\cdot|\Sigma|$ systems of linear equations over $\ZZ$ to determine
membership in the module generated by
$\{ \init \mu(w) : w\in W_k\}$.  Again, this requires time polynomial
in $\mathcal{A}$.  Altogether, the algorithm runs in polynomial time.

If the loop terminates by returning $W\subseteq \Sigma^*$ then
$\{ \init \trans(w) : w\in W\}$ 
contains $\init$ and is closed by multiplication on the right by
$\trans(\sigma)$ for all $\sigma\in \Sigma$.  Thus this module is
the forward module of $\mathcal{A}$.
\end{proof}

The procedure to compute an equivalent $\ZZ$-automaton from a $\QQ$-automaton is illustrated in \Cref{alg:Z_compute}. It starts by computing a $\QQ$-basis of the backward space and by building an equivalent $\QQ$-automaton $\A'$, where  each entry of a forward reachability vector is an evaluation of the function computed by $\A$, that is $\init'\trans'(u) = \begin{bmatrix} \sem{\A}(uw_1)& \ldots & \sem{\A}(uw_m)\end{bmatrix}$. We then apply \ZGenerators{$\A'$} to either deduce the existence of a word such that $\sem{\A}(ww_i) \not\in \ZZ$ for some $i \in \{1,\ldots,m\}$, or to obtain a generator of the forward reachability set consisting of integer vectors. Form these generators, an equivalent $\ZZ$-automaton is built.


\begin{algorithm}[t]
\Fn{\ZCompute{$\A = (\init,\trans,\fin)$}}{ 
  \tcp{\color{blue}Compute a basis of the backward space}
  $W_B := \{ \varepsilon\}$\;\label{alg_line-start_tzeng}
  \SetAlgoVlined
  \SetAlgoShortEnd
  \While{there is $(w,\sigma) \in W_B\times \Sigma$ s.t. $\trans(\sigma w)\fin \not\in \mspan[\QQ]{\trans(u)\fin : u \in W_B}$}{
    
    $W_B := W_B \cup \{\sigma w\}$
  }\label{alg_line-end_tzeng}
  \SetAlgoNoEnd
  \SetAlgoNoLine
  \tcp{\color{blue}Build new automaton}
  $B := \begin{bmatrix} \trans(w_1)\fin & \ldots & \trans(w_m)\fin\end{bmatrix}$ where $W_B = \{ w_1,\ldots,w_m\}$\;
  \tcp{\color{blue}Conjugate $\A$ with matrix $B$ to obtain $\A'$}
  $\A' := (\init',\trans',\fin')$ s.t. $\init' = \init B$, $B\fin' = \fin$ and $B\trans'(\sigma) = \trans(\sigma)B$, for all $\sigma \in \Sigma$\;
  \Switch{\ZGenerators{$\A'$}}{
    \Case(\tcp*[f]{\color{blue}$\init'\trans'(w) \not\in \ZZ^m$}){$w \in \Sigma^*$}{
      take $i \in \{1,\ldots,m\}$ such that $(\init'\trans'(w))_i \not\in \ZZ$\;
      \KwRet{$ww_i$}
    }
    \Case(\tcp*[f]{\color{blue}Generators of forward space in $\ZZ^m$}){$W \subseteq \Sigma^*$}{
      $B_F :=$ generate $\ZZ$-basis of $\mspan[\ZZ]{\init'\trans'(w) \mid: w \in W}$
      \tcp*{\color{blue}Using Smith Normal Form}
      $F := \begin{bmatrix} \vv_1 \\ \ldots \\ \vv_\ell\end{bmatrix}$ where $B_F = \{\vv_1,\ldots,\vv_\ell\}$\;
      \tcp{\color{blue}Conjugate $\A'$ with matrix $F$ to obtain $\A''$}
      $\A'' := (\init'',\trans'',\fin')$ s.t. $\init''F = \init'$, $\fin'' = F\fin'$ and $\trans''(\sigma)F = F\trans'(\sigma)$, for all $\sigma \in \Sigma$\;
      \KwRet{$\A''$}
    }
  }
}
\caption{Computing a $\ZZ$-weighted automaton from a $\QQ$-weighted automaton.}
\label{alg:Z_compute}
\end{algorithm}

\begin{theorem}
The procedure  \ZCompute, described in \Cref{alg:Z_compute}, is a polynomial-time algorithm that given a $\QQ$-weighted automaton $\A$ of dimension $n$ over $\Sigma$, either outputs an  
equivalent $\ZZ$-automaton (that is in fact minimal as a  
$\QQ$-weighted automaton), or a word $w$ such that  
$\sem{\A}(w)\not\in\mathbb{Z}$.
\label{thm:main}
\end{theorem}

\begin{proof}
  The procedure is a variant of  
  the classical minimisation algorithm for weighted automata over fields.
  
  The first step is to compute a basis $\{ u_1,\ldots,u_m\}$ of the  backward space of
  $\A$.  Lines \ref{alg_line-start_tzeng}-\ref{alg_line-end_tzeng} correspond to Tzeng's procedure and, as noted previously, this is done in polynomial time.
  The matrix~$B\in\mathbb{Q}^{n\times m}$ has columns corresponding to the vectors in the
  above-mentioned basis, that are $\trans(u_i)\fin$. 
  
  The next step defines a new $m$ dimensional $\QQ$-automaton $\A'$ that is a conjugate of $\A$, so that $\sem{\A} = \sem{\A'}$. 
  From the fact that the columns of $B$ form a basis of
  the backward space of $\A$ it can be seen that $\A'$
  is well-defined. 
  Furthermore, for all $w\in\Sigma^*$ we have
  $\valpha' \mu'(w) = \valpha \mu(w) B$, so, the $i$-th entry of
  $\valpha' \mu'(w)$ has the form $\valpha \mu(ww_i) \veta= \sem{\A}(ww_i)$. Thus, the forward reachability set of $\mathcal{A}'$
  consists exclusively of integer vectors when $\sem{\A}$
  is integer-valued.
  
  Applying \Cref{prop:promise}, the computation of \ZGenerators{$\A'$} yields either a word~$w\in \Sigma^*$ such that $\init\trans(w) \not\in \ZZ^m$ or else a set~$W$ of words generating the forward reachability set of~$\A'$. In the former case, there exists $i \in \{1,\ldots,m\}$ such that $(\init\trans(w))_i \not\in \ZZ$ and so $\sem{\A}(ww_i) \not\in \ZZ$. In the latter case, we use the Smith Normal Form to generate a $\ZZ$-basis $B_F$ of $\mspan[\ZZ]{\init'\trans'(w) : w \in W}$. As $B_F$ is comprised of the $\ZZ$-vectors $\vv_1,\ldots,\vv_\ell$, the $\ell$ dimensional automaton $\A''$ is a conjugate automaton of $\A'$, that is $\sem{\A'} = \sem{\A''}$. Note that $\A''$ is
  a well-defined $\ZZ$-automaton by the fact that the rows of $F$
  form a $\ZZ$-basis forward module of $\A'$, which entails that Equation~\eqref{eq:forward} has a solution $\alpha_f,\mu_f(\sigma),\beta_f$ in integers. We conclude by noting that $\ell$ is  the dimension of the forward space of
  $\A'$ as well as the rank of the forward module.  It follows
  that $\A''$ is a minimal $\QQ$-weighted automaton.
\end{proof}


\subsection{Exact Learning}

In this subsection, we describe how the exact learning problem for $\ZZ$-weighted automata can be reduced to the exact learning problem for
 $\QQ$-weighted automata.  Such a reduction is non-trivial since the equivalence oracle in the former setting is more restrictive: it requires a $\ZZ$-weighted automaton as input rather than a $\QQ$-weighted automaton.
 The key to the reduction is thus a procedure \QEquiv that implements 
 an equivalence oracle for $\QQ$-weighted automata using an equivalence oracle for 
 $\ZZ$-weighted automata.
This procedure inputs 
a $\QQ$-weighted automaton $\AH$ and returns either \KwSome{$w$} or \KwNone:
\begin{itemize}
    \item In the first case, it returns \KwSome{$w$} 
with $w$ being a counterexample, witnessing that $\sem{\A}(w) \neq \sem{\AH}(w)$. This counterexample is given by \ZCompute{$\AH$} 
in case $\sem{\AH}$ is not integer valued and otherwise it is
given by \Equiv. 
\item In the second case the procedure returns \KwNone, meaning  that  $\AH$ is equivalent to $\A$.
\end{itemize}  


\begin{algorithm}[t]
\Fn{\QEquiv{$\AH$}}{
  \Switch{\ZCompute{$\AH$}}{
    \lCase(\tcp*[f]{\color{blue}A counterexample as $f_\AH(w) \not\in \ZZ$}){$w \in \Sigma^*$}{\KwRet{\KwSome{w}}}
    \lCase(\tcp*[f]{\color{blue}$\ZZ$-automaton equivalent to $\AH$}){$\AH'$}{
      \KwRet{\Equiv{$\AH'$}}
    }
  }
}
\end{algorithm}


\begin{theorem}
There is a procedure that learns the target $\ZZ$-weighted automaton $\A$, by outputting a minimal $\ZZ$-weighted automaton equivalent to $\A$, which runs in  polynomial time  in the length of the encoding of $\A$ and in the length of the longest counterexample given by the teacher. 
\end{theorem}

\begin{proof}
Denote by $s$ the size of the encoding of the target automaton~$\A$. As is the case for $\Q$-weighted automata learning, the algorithm maintains the invariant that the dimension of the hypothesis automata~$\AH$ constructed during the learning procedure is less than~$s$.  
By \Cref{thm:main}, the procedure \ZCompute{$\AH$} runs in time polynomial  in~$s$.  This implies that  the built-in \QEquiv{$\AH$} also runs in time polynomial  in~$s$. 
We know that there is a procedure~$\mathcal{L}$ that learns $\QQ$-weighted automata, and runs in  time polynomial in $s$ and in the length of the longest counterexample given by the teacher~\cite{exact-learning-wa}. 
As such, $\mathcal{L}$ only calls such equivalence oracle a polynomial number of times. Therefore, using \QEquiv as an oracle for~$\mathcal{L}$ yields a polynomial time procedure that outputs a $\QQ$-weighted automaton $\AH$ equivalent to $\A$. We conclude by calling \ZCompute{$\AH$} which runs, as already mentioned, in time polynomial in~$s$. 
\end{proof}

%% file: holonomic_automata.tex


 Recall that a P-finite automaton of dimension $n$ over $\Sigma$ is a tuple $\A = (\init,\trans,\fin(x))$ where $\init \in \QQ^{1\times n}$ is the initial vector, $\trans: \Sigma \rightarrow \QQ[x]^{n\times n}$ is the transition function and $\fin(x) \in \QQ[x]^{n\times 1}$ is the final vector.  We write $\trans(\sigma,k)$ to stand for $\trans(\sigma)(k)$ for all $\sigma \in \Sigma$ and $k \in \NN$. We extend $\trans$ to a map $\trans: \Sigma^* \rightarrow \QQ[x]^{n\times n}$ by writing $\trans(\varepsilon)(x) := I_n$ and 
 \[\trans(w\sigma,x) := \trans(w,x)\, \trans(\sigma,x+|w|)\] for all $\sigma \in \Sigma$ and $w \in \Sigma^*$. Hence, the semantics of $\A$, defined as \[
 \sem{\A}(w) = \init\trans(\sigma_1,1)\ldots\trans(\sigma_k,k)\fin(k+1)
 \] for all $w = \sigma_1\ldots\sigma_k \in \Sigma^*$, can be simply written $\sem{\A}(w) = \init\trans(w,1)\fin(|w|+1)$. The semantics of~$\A$ is also called the function computed by $\A$. We also denote by $\vec{e_1},\vec{e_2},\ldots,\vec{e_n}$ the standard basis.

\medskip

In this section, we tackle the zeroness, equivalence, and exact learning problems for P-finite automata.
The equivalence problem is the problem of deciding whether two  automata compute the same function,
while the zeroness problem aims to check whether the input automaton computes the zero function.
In \Cref{subsec:Hequivalence} we observe that the zeroness and equivalence problems for P-finite automata are polynomial-time interreducible and we show that zeroness can be solved in polynomial time.  Meanwhile, 
in \Cref{subsec:Hlearning} we show that the 
P-finite automata
can be exactly learned in polynomial time in the MAT model.

%% file: equivalence.tex
We can reduce the equivalence problem to the zeroness problem. Indeed, two automata $\A_1$ and $\A_2$ are equivalent if and only if the difference automaton $\A_-$ (such that $\sem{\A_-}=\sem{\mathcal{A}_1}-\sem{\mathcal{A}_2}$) computes the zero function.
We refer to \Cref{sec:proof holonomic equiv} for details.

\begin{restatable}{proposition}{lemmareductionequivtozero}
The equivalence problem of P-finite automata is polynomial-time reducible to the zeroness problem.
\label{lemma: eq_zero}
\end{restatable}



\subsubsection{Backward module}
Below, we  fix a P-finite automaton $\mathcal{A}=(\init, \trans, \fin(x))$   of dimension $n$ over $\Sigma$. 
The \emph{backward function} associated to $\A$, denoted by $\backF{\A}$, is the function $\backF{\A}:\Sigma^*\rightarrow\QQ[x]^n$ given by 
\[\backF{\A}(u)(x) = \trans(u,x)\fin(x+|u|)\,.\]
The \emph{backward module} is the $\QQ[x]$-submodule of $\QQ[x]^n$ defined as $\backS{\A} = \mspan[{\QQ[x]}]{\backF{\A}(w) : w \in \Sigma^*}$.

Consider the  P-finite automaton of Program~\ref{prog:fact}, one can show that the backward module $\backS{\A_1}$ of this automaton  is defined as:
\[
\mspan[{\QQ[x]}]{ 
\begin{bmatrix}0\\x\end{bmatrix},
\begin{bmatrix}0\\xp_k(x)\end{bmatrix},
\begin{bmatrix}xp_k(x)\\0\end{bmatrix}
: k\in \NN},
\]
where $p_k(x):=\prod_{i=1}^k (x+i)$. By a simple computation, we have that  
\[\backS{\A_1} = \mspan[{\QQ[x]}]{\begin{bmatrix}0\\x\end{bmatrix},
\begin{bmatrix}x(x+1)\\0\end{bmatrix}} \, . \]


We remark that the backward function can be defined recursively as 
$\backF{\A}(\varepsilon)=\fin(x)$, and for all $\sigma \in \Sigma$ and  $w\in\Sigma^*$, 
\[\backF{\A}(\sigma w)=\trans(\sigma,x)\, \backF{\A}(w)(x+1),\]
where $\backF{\A}(w)(x+1)$ is  obtained by substituting $x+1$ for $x$ in the vector~$\backF{\A}(w)$. 
By definition, the result of the computation of a P-finite automaton $\A$ on a word $w$ is $\sem{\A}(w)=\init\backF{\A}(w)(1)$. 


\paragraph{From backward module to zeroness} Formally speaking, the zeroness problem asks, given an automaton~$\A$ over~$\Sigma$, whether $\sem{\A}(w) = 0$ for  all words~$w \in \Sigma^*$. 
The following proposition describes how we can decide zeroness by inspecting a finite generating set of the backward module.


\begin{proposition}
\label{lemma:zero_ha_bcw}
Let $\A = (\init,\trans,\fin(x))$ be a P-finite automaton of dimension $n$. Let $S \subseteq \mathbb{Q}[x]^n$ be a finite generating set for the backward module $\backS{\A}$. We have $\sem{\A} \equiv 0$ if and only if  all $\boldsymbol{v} \in S$ satisfy~$\init v(1) = 0$.
\end{proposition}

\begin{proof}

Let $\A$ be over $\Sigma$. The proof is straightforward by unfolding the definitions of backward function, backward module, and $f_\A$:
\begin{align*}
    \forall w\in \Sigma^*, \sem{\A}(w)=0 &\iff \forall w\in\Sigma^*: \init B_{\mathcal{A}}(w)(1) = 0\\
    &\iff \forall v \in \backS{\A}, \init v(1) = 0 &\text{By the definition of $\backS{\A}$}\\
    &\iff \forall v \in S, \init v(1) = 0 &\text{Since $S$ is a generating set of $\backS{\A}$}
\end{align*}
\end{proof}

The previous proposition indicates that, in order to verify zeroness, it is enough to check if $\init$ is orthogonal to a generating set of the backward module. 


\subsubsection{Computing a generating set for \texorpdfstring{$\backS{\A}$}{BA}}
Our algorithm for computing a generating set of the backward module is displayed in \Cref{alg:backward space generation}. It bears a strong resemblance to our algorithm for computing generators for the backward and forward modules in $\ZZ$-weighted automata (\Cref{alg:forward module}). The main distinction  lies  in the soundness proof, which is more involved due to the necessity to work with $\QQ[x]$-modules. 
In particular, we will need the following corollary of \Cref{lemma:general_length}.

\SetKwFunction{GenBackward}{generators\_backward\_module}
\SetKwFunction{GenBackward}{generators\_backward\_module}

\begin{algorithm}[t]

\Fn{\GenBackward{$\A$}}{
  \SetAlgoVlined
  \SetAlgoShortEnd
  
  \tcp{\color{blue}$\A$ a P-finite automaton over $\Sigma$}
  $W := \{\varepsilon\}$\;

  \tcp{\textcolor{blue}{Finding words that increase the rank}}
  \While{there is $(w,\sigma) \in W\times \Sigma$ such that $\backF{\A}(\sigma w) \not\in \mspan[\QQ(x)]{\backF{\A}(u) : u \in W}$}{\label{backward_space-while1}
    $W := W \cup \{\sigma w\}$
  }

  \tcp{\textcolor{blue}{Finding words that augment the module}}
  \While{there is $(w,\sigma) \in W\times \Sigma$ such that $\backF{\A}(\sigma w) \not\in \mspan[{\QQ[x]}]{\backF{\A}(u) : u \in W}$}{\label{backward_space-while2}
    $W := W \cup \{\sigma w\}$
  }
  \KwRet{$\{\backF{\A}(u) : u \in W\}$}
}
 
\caption{Finding a generating set for the backward module of a P-finite automaton} 
\label{alg:backward space generation}
\end{algorithm}

\begin{corollary}
\label{cor:sequence of increasing modules}
Let $n,k \in \NN$ and $M_0 \subsetneq M_1 \subsetneq \ldots \subsetneq M_k$ be a strictly increasing chain of submodules of $\QQ[x]^n$, all having the same rank $r \leq n$. Assume that $M_0$ is generated by a collection of vectors whose entries have degree at most $d$. Then $k \leq d \cdot r$.
\end{corollary}

\begin{proof}
From \Cref{lemma:general_length}, it follows that $k$ is bounded by the number of prime factors of $D_r(A)$ where $A$ is the matrix whose columns contain generators of $M_0$. Since the number of prime factors of a univariate polynomial  is at most its degree, $k$ is bounded by $\deg(D_r(A))$. This can also be upper-bounded by the maximum degree of all $r\times r$ minors of $A$, which, by the triangle inequality and the determinant formula involving permutations, is at most $d\cdot r$.
\end{proof}

We are now ready to  present the  polynomial-time membership of the   equivalence problem of P-finite automata.


\begin{theorem}
\label{th:backward generation}
 The procedure \GenBackward in \Cref{alg:backward space generation}, on an input   P-finite automaton $\A$, terminates and outputs a set $B$ of  vectors such that $\backS{\A} = \mspan[{\QQ[x]}]{B}$. The procedure executes in polynomial time in the length of encoding of~$\A$.
\end{theorem}

\begin{proof}
Let $\A = (\init,\trans,\fin(x))$ be an automaton  of dimension $n$ and over alphabet $\Sigma$. Write $W_1,W_2,\ldots$ for the successive instantiations of the variable $W$ during the execution of the function \GenBackward{$\A$}.
Since  $W_1 = \{ \varepsilon\}$, and for all $i > 0$, $W_i=W_{i-1} \cup \{ \sigma w\}$ for some $\sigma \in \Sigma$ and $w \in W_{i-1}$, it follows that 
the maximum length of words in $W_i$ is at most the size of~$W_i$. 

The first {\bf while}-loop, in Line \ref{backward_space-while1},
terminates after at most~$n$ iterations since the backward module, being a submodule of  $\QQ[x]^n$, has rank at most $n$.
The second  {\bf while}-loop, in Line \ref{backward_space-while2},
terminates by virtue of $\QQ[x]^n$ being  Noetherian. Below, we write
\[W_\ell = \{ w_1,\ldots,w_\ell\} \qquad \text{ and } \qquad W_m = \{ w_1,\ldots,w_m\} \]
for some $\ell\leq n$,  for the instantiations of $W$ upon exiting the first and second {\bf while}-loops, respectively.

We first claim that $\backF{\A}(w) \in \mspan[\QQ(x)]{\backF{\A}(u) : u \in W_\ell}$ for all words~$w \in \Sigma^*$.
The proof is by induction on the length of the words. The base case ($|w| = 0$) follows as~$\varepsilon \in W_\ell$.  
For the  inductive step ($|w| > 0$), decompose $w$ as  $\sigma w'$ for some $\sigma \in \Sigma$ and $w' \in \Sigma^*$. By the induction hypothesis, 
\[\backF{\A}(w')(x) = \sum_{k=1}^\ell \frac{p_{k}(x)}{q_{k}(x)} \backF{\A}(w_k)(x)\] for some univariate polynomials $p_{k}(x),\, q_{k}(x)\in \QQ[x]$, where $k\in \{1,\ldots,\ell\}$.
Recall the recursive definition of the backward function, namely, we have $\backF{\A}(\sigma w')=\trans(\sigma,x)\, \backF{\A}(w')(x+1)$. Hence, 
\begin{align*}
      \backF{\A}(\sigma w')(x) = & \trans(\sigma,x) \sum_{k=1}^\ell \frac{p_{k}(x+1)}{q_{k}(x+1)} \, \backF{\A}(w_k)(x+1) \\
      = & \sum_{k=1}^\ell \frac{p_{k}(x+1)}{q_{k}(x+1)} \trans(\sigma,x)\backF{\A}(w_k)(x+1) \\
      = & \sum_{k=1}^\ell \frac{p_{k}(x+1)}{q_{k}(x+1)} \backF{\A}(\sigma w_k)\,,
  \end{align*}
 implying  that $\backF{\A}(w)\in \mspan[\QQ(x)]{\backF{\A}(\sigma u) : \sigma \in \Sigma, u \in W_\ell}$.
 But then 
 the exit-condition of the first {\bf while}-loop  ensures that  \[\mspan[\QQ(x)]{\backF{\A}(\sigma u) : \sigma \in \Sigma, u \in W_\ell} \subseteq \mspan[\QQ(x)]{\backF{\A}( u) :   u \in W_\ell},\] 
concluding the proof of the claim.

We show a similar result concerning the second {\bf while}-loop termination.
We  claim that $\backF{\A}(w) \in \mspan[{\QQ[x]}]{\backF{\A}(u) : u \in W_m}$ 
for all words~$w \in \Sigma^*$. 
Intuitively speaking, once exiting the second {\bf while}-loop, no words in $\Sigma^*$ that could augment the module can be added. 
The proof is again by induction on the length of the words.  The base case ($|w| = 0$) trivially holds as~$\varepsilon \in W_m$.
For the  inductive step ($|w| > 0$), rewrite $w$ as  $\sigma w'$ for some $\sigma \in \Sigma$ and $w' \in \Sigma^*$. By the induction hypothesis, 
\[\backF{\A}(w')(x) = \sum_{k=1}^m p_{k}(x) \backF{\A}(w_k)(x)\]
for some polynomials $p_k(x) \in \QQ[x]$ where $k\in \{1,\ldots,m\}$. Following similar reasoning as in the first loop case,
 we obtain that
 $\backF{\A}(w)\in \mspan[{\QQ[x]}]{\backF{\A}(\sigma u) : \sigma \in \Sigma, u \in W_m}$.
But then, again,
 the exit-condition of the second {\bf while}-loop  ensures that  \[\mspan[{\QQ[x]}]{\backF{\A}(\sigma u) : \sigma \in \Sigma, u \in W_m} \subseteq \mspan[{\QQ[x]}]{\backF{\A}( u) :   u \in W_m},\] 
concluding the proof of the claim.

It remains to show that the execution of \GenBackward{$\A$} can be carried out in time polynomial in the length of encoding of~$\A$.  
Recall that, given a word  $w = \sigma_1\ldots\sigma_k$, the backward reachable vector is computed as $\backF{\A}(w)(x) = \trans(\sigma_1,x)\ldots\trans(\sigma_k,x+k-1)\fin(x+k)$.
Denote by $d$ and $c$, respectively, the maximal degree and largest coefficient of the polynomials occurring as entries of~$\fin(x)$ and $\trans(\sigma)$, for~$\sigma \in \Sigma$.
It follows that the degree of polynomial entries of~$\backF{\A}(w)(x)$  is at most~$d(|w|+1)$. 
We will argue that  the largest coefficient of the polynomials occurring as entries of~$\backF{\A}(w)(x)$ is at most~$n^{|w|}c^{|w|+1}(|w|d)^{(|w|+1)d}$. Indeed, this comes from the  observation that  the coefficients of the monomial $(x+|w|)^d$ are bounded by $(|w|d)^d$. 
The length of the encoding of  $\backF{\A}(w)(x)$ is therefore polynomial in the length of encoding of~$\A$ and in~$|w|$.
  Using \cite{KANNAN198569}, we deduce that testing whether 
  \[\backF{\A}(\sigma w) \not\in \mspan[\QQ(x)]{\backF{\A}(u) : u \in W}\qquad \text{ or } \qquad \backF{\A}(\sigma w) \not\in \mspan[{\QQ[x]}]{\backF{\A}(u) : u \in W}\] is polynomial in the length of encoding of~$\A$, and in the maximum length of words in $W_m$, and in the size of~$W_m$. 
Recall that the maximum length of words in $W_m$ is at most the size of~$W_m$.

We conclude the proof by arguing that the size of $W_m$ is polynomial in the length of encoding of~$\A$. As a result of the two claims on the termination of the loops, the two backward modules induced by $W_\ell$ and $W_m$ have the same rank. Since $\ell\leq n$, the degree of the polynomials in the entries of $\backF{\A}(u)$ for $u \in W_\ell$ is at most $d(n+1)$. By \Cref{cor:sequence of increasing modules}, the length of the strictly increasing sequence of modules induced by  $W_{\ell} \subsetneq \ldots \subsetneq W_m$ is at most $m -\ell+1 \leq n d(n+1)$, implying that the size of $W_m$ is at most~$n d(n+1) + n -1$. 
\end{proof}

By a direct application of \Cref{th:backward generation}, \Cref{lemma:zero_ha_bcw} and \Cref{lemma: eq_zero}, we have:
\begin{theorem}
The zeroness and equivalence problems for P-finite automata are both in polynomial time.  We can furthermore suppose that the polynomial-time procedure for testing equivalence returns a word of   
polynomial length that witnesses in-equivalence on negative instances.
\end{theorem}



%% file: learning.tex

\SetKwFunction{LPrefix}{largest\_correct\_prefix}%

We first introduce some notation and terminology. Below, we fix $f: \Sigma^* \rightarrow \QQ$ to be a  function. 
The  Hankel matrix of~$f$ is  an infinite matrix with rows and columns indexed by words in~$\Sigma^*$ such that  $\PHankel{r}{c}:=f(rc)$, where $H(r, c)$ is the entry of matrix with row index $r\in \Sigma^*$ and column index~$c\in \Sigma^*$.

Given  two sequences $\SeqR, \SeqC$ of words from $\Sigma^*$,  denote by $\PHankel{\SeqR}{\SeqC}$ the restriction of the Hankel matrix to the respective sets $\SeqR$ of rows and $\SeqC$ of columns, that is, if $\SeqR = [ r_1,\ldots, r_m ]$ and $\SeqC = [ c_1,\ldots,c_n]$, then $\PHankel{\SeqR}{\SeqC}$ is the $m\times n$ submatrix such that $\PHankel{\SeqR}{\SeqC}_{i,j} = f(r_ic_j)$.  Moreover, given two words~$r,c\in \Sigma^*$ such that $r$ appears in the sequence~$\SeqR$  and $c$ appears in the sequence~$\SeqC$, we write 
$\rowF{\SeqC}{r}$ for the associated row  and $\colF{\SeqR}{c}$ for the associated column in 
$\PHankel{\SeqR}{\SeqC}$, namely, 
\[\rowF{\SeqC}{r}:=\begin{bmatrix} f(rc_1) & \ldots & f(rc_n)\end{bmatrix} \qquad \text{ and } \qquad \colF{\SeqR}{c}:=\begin{bmatrix} f(r_1c) & \ldots & f(r_mc) \end{bmatrix}^\top.\]
In the sequel, we will call $\PHankel{\SeqR}{\SeqC}$ a \emph{table}.

Assume that the target function~$f$ can be computed by a P-finite automaton. 
Intuitively, our learning algorithm maintains a table from which it constructs a \emph{hypothesis automaton}. Using the equivalence oracle, the algorithm checks whether the hypothesis automaton computes the function $f$. In case of a negative answer, the witness of non-equivalence is used to augment the table (by augmenting the sets of rows and columns), and the process repeats.

In order to build the hypothesis automaton we require the table to be \emph{closed} in the following sense. Let $\SeqR$ and $\SeqC$ be two sequences of words from $\Sigma^*$ such that $|\SeqC|=n$. We say that the table~$\PHankel{\SeqR}{\SeqC}$ is \emph{closed} when for each $\sigma \in \Sigma$, there exists a matrix of polynomials $M_\sigma(x) \in \QQ[x]^{n\times n}$ such that for all rows~$r \in \SeqR$, the equation
$\rowF{\SeqC}{r\sigma} = \rowF{\SeqC}{r} M_\sigma(|r|+1)$ holds. Given such a closed table~$\PHankel{\SeqR}{\SeqC}$, we can compute 
a hypothesis P-finite automaton  $(\init,\trans,\fin)$ of dimension $n$ as follows:
\begin{equation}
\label{eq:constHA}
    \init = \rowF{\SeqC}{r_1}, \qquad  \fin(x) = \vec{e_1},\qquad  \trans(\sigma,x) = M_\sigma(x) \text{ for all } \sigma \in \Sigma.
\end{equation}
The polynomials in the transition matrix of a hypothesis automaton need to be constructed by interpolation.
To this end, we maintain a variable $d$ that represents a degree bound on the 
polynomials in~$M_\sigma(x)$.  
Specifically, we will say that the table~$\PHankel{\SeqR}{\SeqC}$ is \emph{$d$-closed} when the maximal degree of the polynomials in the $M_\sigma(x)$ are bounded by~$d$. 
The $d$-closedness condition allows  to set up a linear system of equations where the unknowns $y_{i,j,k}$ are the coefficients  of the polynomials of each entry of~$M_{\sigma}(x)$, that is, we write the $(i,j)$-th entry of $M_{\sigma}(x)$ as $y_{i,j,d}x^d+y_{i,j,d-1}x^{d-1}+\cdots+y_{i,j,0}$. More precisely, we search for the unknowns $y_{i,j,k}$, ranging over~$\QQ$. Focusing on the $j$-th column of~$M_{\sigma}(x)$, the $d$-closedness condition $\rowF{\SeqC}{r\sigma} = \rowF{\SeqC}{r} M_\sigma(|r|+1)$ entails, for all $r \in \SeqR$, the following equation:
\begin{equation}
\label{eq:beforeLS}
    f(r \sigma c_j) = \rowF{\SeqC}{r} \sum_{k=0}^d (|r|+1)^k Y_k = \sum_{k=0}^d (|r|+1)^k \rowF{\SeqC}{r} Y_k, 
\end{equation}
where the $Y_k$ are the column vectors $\begin{bmatrix}y_{1,j,k} & \ldots & y_{|\SeqR|,j,k}\end{bmatrix}^\top$ of unknowns. By taking $H = \PHankel{\SeqR}{\SeqC}$ and $\Delta$ the $m\times m$ diagonal matrix $diag(|r_1|+1, \ldots, |r_m|+1)$ where $\SeqR = [r_1,\ldots,r_m]$, we obtain the following system of equations in $Y_0,\ldots,Y_d$:
\begin{equation} 
\label{eq:bigLS}
    \begin{bmatrix}f(r_1\sigma c_j)\\ \vdots \\ f(r_m\sigma c_j) \end{bmatrix} = 
\sum_{k=0}^d \Delta^k H Y_k = 
\begin{bmatrix}\Delta^0 H& \ldots & \Delta^d H \end{bmatrix} 
\begin{bmatrix}Y_0\\ \vdots \\ Y_d 
\end{bmatrix}. 
\end{equation}

We recover the \emph{hypothesis automaton associated to the $d$-closed table~$\PHankel{\SeqR}{\SeqC}$} from a solution to the above system of equations
by setting the $j$-th column of $M_\sigma(x)$ to be $\sum_{k=0}^d  Y_k x^{k}$. 
Henceforth we denote by $\AHankel{\SeqR}{\SeqC}{d}$ the matrix 
\[
  \begin{bmatrix} 
    \Delta^0 H & \ldots & \Delta^d H
  \end{bmatrix}.
\]

In the following proposition, we state a sufficient condition for 
the above linear system of equations to  have a solution, meaning that \emph{the table $\PHankel{\SeqR}{\SeqC}$ is $d$-closed}.

\begin{proposition}
  \label{lem:full rowrank and closeness}
  Given two sequences of words $\SeqR$ and $\SeqC$ and $d \in \NN$,  the table~$\PHankel{\SeqR}{\SeqC}$ is $d$-closed if  $\AHankel{\SeqR}{\SeqC}{d}$ has full row rank.
\end{proposition}
  
  \begin{proof}
  Let  $\SeqR = [r_1,\ldots,r_m]$ and $\SeqC = [c_1,\ldots,c_n]$. Write $A$ for $\AHankel{\SeqR}{\SeqC}{d} \in \QQ^{m\times (d+1)n}$, which, by  hypothesis, has full row rank. 
  Then for all vectors $V \in \QQ^{m\times 1}$, the system $AX = V$ has a solution $X \in \QQ^{(d+1)n\times 1}$. Indeed, the system $AX = V$ has a solution if and only if $\rank{A} = \rank{\begin{bmatrix} A & V \end{bmatrix}}$. Since $A \in \QQ^{m \times (d+1)n}$ and $\rank{A} = m$, we deduce that for all $V \in \QQ^{m\times 1}$, the equality $\rank{\begin{bmatrix} A & V \end{bmatrix}} = \rank{A}$ holds and the system $AX = V$ has a solution. 

Write~$n$ for the size of the sequence~$C$. 
We construct the matrices~$M_{\sigma} \in \QQ[x]^{n\times n}$, for $\sigma\in \Sigma$, as follows. By the above argument, for $j\in \{1,\cdots, n\}$,
the  system of linear equations described in \eqref{eq:bigLS} has some solution, say $Y^*_0,\ldots, Y^*_d$. We define   the $j$-th column of $M_\sigma(x)$ to be $\sum_{k=0}^d  Y^*_k x^{k}$, which in turn implies that  the $(i,j)$-th entry of $M_{\sigma}(x)$ is the polynomial $y^*_{i,j,d}x^d+y^*_{i,j,d-1}x^{d-1}+\cdots+y^*_{i,j,0}$ of degree~$d$.

It remains to argue that the matrix $M_{\sigma}(x)$ so defined satisfies the closedness condition, that is, for all rows~$r \in \SeqR$ the condition
$\rowF{\SeqC}{r\sigma} = \rowF{\SeqC}{r} M_\sigma(|r|+1)$  holds. But then, this  is guaranteed by  enforcing~\eqref{eq:beforeLS} for all columns~$c\in \SeqC$. We conclude by noting that constraints~\eqref{eq:beforeLS} constitute the system of linear equations described in \eqref{eq:bigLS}. 
\end{proof}

\begin{algorithm}[t]
  \Fn{\HBuild{$d,\SeqR,\SeqC = [c_1,\ldots,c_n]$}}{
    \SetAlgoVlined
    \SetAlgoShortEnd
    \tcp{\color{blue}$\PHankel{\SeqR}{\SeqC}$ is assumed to be $d$-closed}
    \For{$\sigma \in \Sigma$}{
      $\SeqC' := [\sigma c_1,\ldots,\sigma c_n]$\;
      solve $\AHankel{\SeqR}{\SeqC}{d} Y = \PHankel{\SeqR}{\SeqC'}$ \tcp*{\color{blue}Has a solution since $\PHankel{\SeqR}{\SeqC}$ is $d$-closed}
      define $\trans(\sigma)_{i,j}(x) := \sum_{k=0}^d Y_{i+kn,j} \, x^k$ for all $i \in \{1,\ldots,n\}$, $j \in \{1,\ldots,n\}$
    }
    \KwRet{$(\vec{e_1}^\top H,\trans,\vec{e_1})$}
  }
\caption{Building an associated P-finite automaton to $\PHankel{\SeqR}{\SeqC}$}\label{algo:buildHaut}
\end{algorithm}

Our algorithm for building the automaton associated to the $d$-closed table~$\PHankel{\SeqR}{\SeqC}$ is given as function \HBuild  in \Cref{algo:buildHaut}. This function is an implementation of the construction stated in~\eqref{eq:constHA}, which is  ensured by the $d$-closedness assumption  on the input table. We note again that the maximal degree of polynomials in constructed automaton is at most~$d$.  In summary, we have: 


\begin{corollary}
    \label{col:functionHbuild}
 The function \HBuild{$d,\SeqR,\SeqC$}, assuming that $\PHankel{\SeqR}{\SeqC}$ is $d$-closed, outputs an automaton $\AH$ associated to the $d$-closed table $\PHankel{\SeqR}{\SeqC}$. 
\end{corollary}

Concretely, in the function \HBuild computing $\AHankel{\SeqR}{\SeqC}{d}$ and $\PHankel{\SeqR}{\SeqC}$ can be evaluated by asking  membership queries from the teacher, through \Membership, at most $|\SeqR|\times|\SeqC|$ times. Therefore, the execution of this function  runs in time polynomial  in $d+|\SeqR|+|\SeqC|+|\Sigma|$.


\subsubsection{Correctness of  P-finite automata}


Let $\SeqR = [ r_1,\ldots,r_m]$ and $\SeqC = [c_1,\ldots,c_n]$ be two sequences of words from $\Sigma^*$. Assume that $\PHankel{\SeqR}{\SeqC}$ is closed and let $\AH = (\init,\trans,\fin)$ be an associated P-finite automaton over $\Sigma$. We say that that \emph{$\AH$ is \emph{correct} on the word~$w \in \Sigma^*$} if $\init\trans(w,1) = \rowF{\SeqC}{w}$.

As previously mentioned, after building a hypothesis automaton $\AH$ associated with a table, we will ask the teacher  an equivalence query on~$\AH$, through   \Equiv{$\AH$}, and receive  a counterexample $w$ in case $\AH$ is not equivalent to the target automaton. The automaton~$\AH$ is correct on $\varepsilon$ by construction and is necessarily incorrect on $w$, as indeed we initialize $\SeqC$ with $\varepsilon$ and ensure that the automaton is always correct on this word,  and the fact that $f(w) \neq \sem{\AH}(\varepsilon \cdot w) = \init\trans(w,1)\vec{e_1}$. We compute the longest prefix $u\sigma$ of $w$ such that $\AH$ is correct on $u$ but incorrect on $u\sigma$. Our learning algorithm extends its table by adding the row associated with $u$.

Computing such a prefix $u$ can be straightforwardly done as depicted in \Cref{alg:largest prefix}. The function \LPrefix{$\AH,\SeqC,w$} outputs $u,\sigma$ as well as the word $c_j \in \SeqC$ that renders $\AH$ incorrect on $u\sigma$.
The execution of \LPrefix{$\AH,\SeqC,w$} runs in time polynomial in its parameters, that is, in time polynomial in $|w|+|\SeqC|$.

In \Cref{col:functionHbuild}, we assumed that the table $\PHankel{\SeqR}{\SeqC}$ is closed in order to build the hypothesis automaton from it. However, when augmenting the table with the row associated with $u$, the closedness condition might  not hold anymore as the new row~$u$ might be linearly dependent with the previous rows in~$\SeqR$. 
We show in the next proposition that in such cases the closedness of the table can be restored by adding the column associated with $\sigma c_j$ to the table, where  $(u,\sigma,c_j)$ is the output of the function~\LPrefix{$\AH,\SeqC,w$}.

\begin{algorithm}[t]
  \Fn{\LPrefix{$\AH,\SeqC = [c_1,\ldots,c_n],w$}}{
    \SetAlgoVlined
    \SetAlgoShortEnd
    \tcp{\color{blue}$\AH = (\init,\trans,\fin(x))$ is a P-finite automaton of dimension $n$ over $\Sigma$}
    $\vec{x} := \init$, $u := \varepsilon$, $v := w$\;
    \While{$v = \sigma w'$}{
      $\vec{x} := \vec{x}\trans(\sigma,|u|+1)$\;
      \For{$j =1 \ldots n$}{
        $y :=$ \Membership{$u\sigma c_j$}\;
        \lIf{$\vec{x}_j \neq y$}{
          \KwRet{$(u,\sigma,c_j)$}
        }
      }
      $u := u\sigma$, $v := w'$\;
    }
  }
\caption{Largest correct prefix}
\label{alg:largest prefix}
\end{algorithm}

\begin{proposition}
\label{lem:rank}
Let  $d \in \NN$, and    $\SeqR,\SeqC$  be sequences of words such that $\AHankel{\SeqR}{\SeqC}{d}$ has full row rank. Let  
 $\AH = (\init,\trans,\fin)$ be an automaton associated to the $d$-closed table $\PHankel{\SeqR}{\SeqC}$. 
Let $u \in \Sigma^*$ and $\sigma \in \Sigma$ be
such that $\AH$ is correct on $u$ but not on $u\sigma$. Let $c_j$ be the $j$-th word in~$\SeqC$ where $ f(u \sigma c_j) \neq \init\trans(u\sigma,1)\vec{e_j}$.

Define  $\SeqR':=\SeqR \cdot [u]$. Then  the   matrix $\begin{bmatrix} \AHankel{\SeqR'}{\SeqC}{d} & \colF{\SeqR'}{\sigma c_j} \end{bmatrix}$ is  full row rank.
\end{proposition}

\begin{proof}
Write $M$ for the matrix $\begin{bmatrix} \AHankel{\SeqR'}{\SeqC}{d} & \colF{\SeqR'}{\sigma c_j} \end{bmatrix}$. 
Since $\AHankel{\SeqR}{\SeqC}{d}$ is a sub-matrix of $M$ and $M$ has $|\SeqR|+1$ rows, we deduce that $|\SeqR| +1 \geq \rank{M} \geq \rank{\AHankel{\SeqR}{\SeqC}{d}} = |\SeqR|$. For a contradiction, assume that $M$ is not full row rank, implying that~$\rank{M}=|\SeqR|$. Then the last row of $M$ is a linear combination of all other rows of $M$.
In other words, writing $\SeqR$ as the sequence~$[r_1,\ldots,r_m]$, there exists a row vector $\rowvec{x}$ such that for all words~$c \in \SeqC$, for all $k \in \{0,\ldots,d\}$, 
\[
\left\{
\begin{array}{l}
\rowvec{x} 
\begin{bmatrix}
(|r_1|+1)^k f(r_1c)& \ldots & (|r_m|+1)^k f(r_m c)
\end{bmatrix}^\top
= (|u|+1)^k f(uc)\\
\rowvec{x} 
\begin{bmatrix}
f(r_1\sigma c_j)& \ldots & f(r_m \sigma c_j)
\end{bmatrix}^\top
= f(u \sigma c_j)
\end{array}
\right.
.
\]

By \Cref{lem:full rowrank and closeness}, since  $\AHankel{\SeqR}{\SeqC}{d}$ has full row rank,   the table $\PHankel{\SeqR}{\SeqC}$ is $d$-closed, implying that for each row $i \in \{1,\ldots,m\}$, 
the equality 
$\rowF{\SeqC}{r_i\sigma} = \rowF{\SeqC}{r_i}\trans(\sigma,|r_i|+1)$ holds.

Recall that the $\trans(\sigma)$ are matrices of univariate polynomials. 
 Define $\trans^{(k)}_{\sigma}$ 
 to be  the matrix whose $(i,j)$-th entry is the coefficient of  $x^k$ in the  $(i,j)$-th entry of $\trans(\sigma)$.  
We obtain that $\rowF{\SeqC}{r_i\sigma} = \rowF{\SeqC}{r_i} \sum_{k=0}^d (|r_i|+1)^k \trans^{(k)}_{\sigma}$. Therefore,
\[
\begin{bmatrix}
  f(r_1\sigma c_j)\\
  \ldots\\
  f(r_m\sigma c_j)
\end{bmatrix} =
\sum_{k=0}^d 
\begin{bmatrix}
  (|r_1|+1)^k f(r_1 c_1) & \ldots & (|r_1|+1)^k f(r_1 c_n) \\
  \ldots&\ldots&\ldots\\
  (|r_m|+1)^k f(r_m c_1) & \ldots & (|r_m|+1)^k f(r_m c_n) \\
\end{bmatrix}
\trans^{(k)}_{\sigma} \vec{e_j}
.
\]
Multiplying both sides of the equation by $\rowvec{x}$, we obtain:
\[
f(u \sigma c_j) = \sum_{k=0}^d 
\begin{bmatrix}
(|u|+1)^k f(u c_1) & \ldots & (|u|+1)^k f(u c_n)
\end{bmatrix} \trans^{(k)}_{\sigma} \vec{e_j}\,,
\]
which in turn implies that  \[
f(u \sigma c_j) = \rowF{\SeqC}{u} \sum_{k=0}^d (|u|+1)^k \trans^{(k)}_{\sigma} \vec{e_j} = \rowF{\SeqC}{u} \trans(\sigma,|u|+1)\vec{e_j}.
\]

By hypothesis, the automaton $\AH$ is correct on $u$, meaning that  $\init\trans(u,1) = \rowF{\SeqC}{u}$. Subsequently, \[f(u \sigma c_j) = \init\trans(u,1)\trans(\sigma,|u|+1)\vec{e_j} = \init\trans(u\sigma,1)\vec{e_j}.\] This is in contradiction with the assumption $\init\trans(u\sigma,1)\vec{e_j} \neq f(u \sigma c_j)$, concluding the proof.
\end{proof}


\begin{corollary}
  \label{cor:rank}
Let  $d \in \NN$, and    $\SeqR,\SeqC$  be sequences of words such that $\AHankel{\SeqR}{\SeqC}{d}$ is full row rank. Let  
 $\AH = (\init,\trans,\fin)$ be an automaton associated to the $d$-closed table $\PHankel{\SeqR}{\SeqC}$. 
Let $u \in \Sigma^*$ and $\sigma \in \Sigma$ be
such that $\AH$ is correct on $u$ but not on $u\sigma$. Let $c_j$ be the $j$-th word in~$\SeqC$ where $ f(u \sigma c_j) \neq \init\trans(u\sigma,1)\vec{e_j}$.

Define  $\SeqR':=\SeqR \cdot [u]$ and $\SeqC':=\SeqC \cdot [\sigma c_j]$. For all $d'>d$,    the table $\AHankel{\SeqR'}{\SeqC'}{d'}$ has  full row rank.
\end{corollary}
\begin{proof}
    The result follows from \Cref{lem:rank} and the fact that $\begin{bmatrix} \AHankel{\SeqR'}{\SeqC}{d} & \colF{\SeqR'}{\sigma c_j} \end{bmatrix}$ is a submatrix of  $\AHankel{\SeqR'}{\SeqC \cdot [\sigma c_j]}{d'}$. 
\end{proof}

We can combine the above-mentioned functions to define our \emph{partial learner}, which is depicted in \Cref{func:partial learner}. It takes four arguments:  two sequences of words $\SeqR$, $\SeqC$ that determine the table, an integer $d_{max}$ representing our \emph{guess} of the maximal degree of polynomials occurring in the target automaton, and finally a \emph{timeout} integer $\ell$ on the number of columns $|\SeqC|$. This limit~$\ell$ acts as a safeguard in case our guess $d_{max}$ is incorrect.

By construction, when \PLearner{$d_{max},\ell,\SeqR,\SeqC$} returns \KwSome{$\AH$}, the automaton $\AH$ computes the target function $f$, that is $\sem{\AH} = f$. However, the function may return \KwNone if we fail to find an equivalent automaton within the time bound imposed by $\ell$. In the next section, we will show that if we take $\ell$ large enough and if we have guessed correctly the value $d_{max}$ for the maximal degree of polynomials occurring in the target automaton, then \PLearner{$d_{max},\ell,[\varepsilon],[\varepsilon]$} will always eventually learn the target.


\subsubsection{Bounding the number of columns}

The function \PLearner incorporates a limit~$\ell$ on the number of columns added while constructing the table. Having this limit ensures that our algorithm terminates even when the guess of the maximal degree of polynomials $d_{max}$ in the target automaton is incorrect. However, we need also to guarantee that we never exceed the limit when the guess for $d_{max}$ is correct. We can compute such a limit by relating the columns $c_1,\ldots,c_n$ that are added during the execution of \PLearner with the submodule of $\backS{\A}$ generated by \[\backF{\A}(c_1),\, \ldots,\, \backF{\A}(c_n).\] Intuitively, the following proposition shows that the sequence of modules $\mathcal{M}_1,\ldots,\mathcal{M}_n$ defined as $\mathcal{M}_i = \mspan[{\QQ[x]}]{\backF{\A}(c_j) : j\leq i}$, for all $i \in \{1\ldots n\}$, is strictly increasing, that is, $\mathcal{M}_1 \subsetneq \ldots \subsetneq \mathcal{M}_n$.

\medskip

We say that a sequence~$[w_1,\ldots, w_n]$ of words is 
\emph{totally suffix-closed}  if for each $w_i$ all its suffixes~$s$
occur before $w_i$ in the sequence, meaning that 
$s=w_j$ for some $j \leq i$.


\begin{proposition}
  \label{lem:bound_column}
  Let $\SeqR, \SeqC$ be  sequences of words from $\Sigma^*$,
  such that $\SeqC$ is totally suffix-closed.
Let $d_{max}$ be the maximal degree of polynomials  in the target P-finite automaton $\A$.
Let $d \geq d_{max}(|\SeqC|+1)|\SeqC|$, 
let~$c\in \SeqC$ and let $\sigma \in \Sigma$. Assume that 
the matrix $\AHankel{\SeqR}{\SeqC}{d}$ does not have full row rank, but the matrix~$\begin{bmatrix} \AHankel{\SeqR}{\SeqC}{d} & \colF{\SeqR}{\sigma c} \end{bmatrix}$ has full row rank.
Then \[\backF{\A}(\sigma c) \not\in \mspan[{\QQ[x]}]{\backF{\A}(c'):c'\in \SeqC}.\]
  \end{proposition}

\begin{proof}
Write $[r_1,\ldots,r_m]$ for the sequence $\SeqR$ and $[c_1,\ldots,c_n]$ for the sequence $\SeqC$. Towards a contradiction, assume that $\backF{\A}(\sigma c) \in \mspan[{\QQ[x]}]{\backF{\A}(c_1),\ldots,\backF{\A}(c_n)}$. Then there exist polynomials~$p_1,\ldots, p_n \in \QQ[x]$ such that $\backF{\A}(\sigma c) = \sum_{i=1}^n p_i(x) \backF{\A}(c_i)$ and hence the equation 
\begin{equation}
\label{eq:XsolutionKannan}
    \begin{bmatrix}\backF{\A}(c_1) & \ldots & \backF{\A}(c_n)\end{bmatrix}
X =
\backF{\A}(\sigma c)
\end{equation}
has solution $X = \begin{bmatrix}p_1(x) & \ldots & p_n(x)\end{bmatrix}^\top$. 

A simple induction on the length of the words~$w$ gives that the degree of the polynomials in $\backF{\A}(w)$ is at most $d_{max}(|w|+1)$. Since  $\SeqC$ is totally suffix-closed, we deduce that $c_1=\varepsilon$, as well as~$\max(|c_1|,\ldots,|c_n|,|\sigma c|) \leq n$. 
These two  above facts imply that the maximal degree of polynomials in $\backF{\A}(c_1)$, \ldots, $\backF{\A}(c_n)$, $\backF{\A}(\sigma c)$ is at most $d_{max}(n+1)$.
By~\cite[Lemma 2.5]{KANNAN198569}, we can assume that the maximum degree of polynomials~$p_i(x)$ in the solution~$X$ of Equation~\eqref{eq:XsolutionKannan} is at most $n$ times the maximum degree of the polynomials in $\backF{\A}(c_1)$, \ldots, $\backF{\A}(c_n)$, $\backF{\A}(\sigma c)$.
Subsequently, the maximum degree of polynomials~$p_i(x)$ in $X$ is at most $d_{max}(n+1)n$. 

Let $(\init, \trans,\fin)$ be the target automaton~$\A$, and let  $f:=\sem{A}$.
Since $\backF{\A}(\sigma c) = \sum_{i=1}^n p_i (x)\backF{\A}(c_i)$ holds, for   all words~$r \in \SeqR$ we have that
\[\init\trans(r,1) \backF{\A}(\sigma c) = \sum_{i=1}^n p_i(|r|+1) \,  \init\trans(r,1) \backF{\A}(c_i) \,, \]
which, in turn, by the definition of the backward function gives
\begin{equation}
\label{eq:jrc}
    f(r \sigma c) = \sum_{i=1}^n p_i(|r|+1)f(rc_i) \,.
\end{equation}

Since $d\geq d_{max}(n+1)n$, we  write the polynomials~$p_i$ in~$X$ as 
$p_i(x):=p^{(0)}_i + p^{(1)}_i x + \ldots + p^{(d)}_i x^d$
where $p^{(k)}_i$ the coefficient of monomial~$x^k$ in~$p_i(x)$. Substituting this representation into~\eqref{eq:jrc}, we get
\begin{equation}
\label{eq:coeffjrc}
    f(r \sigma c) = \sum_{i=1}^n \sum_{k=0}^d  p_i^{(k)} \, (|r|+1)^k \, f(rc_i).
\end{equation}
We group the coefficient  $p_i^{(k)}$ of the monomial~$x^k$   in all polynomials $p_i(x)$ in a single vector~$P^{(k)}$; formally,  define
$d$ vectors $P^{(0)}, \ldots, P^{(d)}$ such that 
 $P^{(k)} = \begin{bmatrix} p_1^{(k)} & \ldots & p_n^{(k)}
\end{bmatrix}^\top$
for  $k\in \{1,\ldots,d\}$. 
From~\eqref{eq:coeffjrc}
we obtain that: 
\[
  \colF{\SeqR}{\sigma c} =
  \begin{bmatrix}
    f(r_1\sigma c) \\ \vdots \\ f(r_m\sigma c)
  \end{bmatrix}
  = \sum_{k=0}^d \Delta^k H P^{(k)} = \AHankel{\SeqR}{\SeqC}{d}
  \begin{bmatrix}
    P^{0}\\
    \vdots\\
    P^{d}
  \end{bmatrix},
\]
where $H=\PHankel{\SeqR}{\SeqC}$ and $\Delta$ is the $m\times m$ diagonal matrix $diag(|r_1|+1, \ldots, |r_m|+1)$. We deduce that the rank of $\AHankel{\SeqR}{\SeqC}{d}$ is equal to the rank of $\begin{bmatrix} \AHankel{\SeqR}{\SeqC}{d} & \colF{\SeqR}{\sigma c} \end{bmatrix}$. This is in contradiction with the assumption that $\AHankel{\SeqR}{\SeqC}{d}$ is not full row rank, but $\begin{bmatrix} \AHankel{\SeqR}{\SeqC}{d} & \colF{\SeqR}{\sigma c} \end{bmatrix}$ is, concluding the proof. 
\end{proof}

Using \Cref{cor:sequence of increasing modules}, we can compute an upper bound on the maximum size of increasing submodule of the backward module $\backS{\A}$.

\begin{algorithm}[t]
  \Fn{\PLearner{$d_{max},\ell,\SeqR,\SeqC$}}{ 
    \lIf{$\ell < |\SeqC|$}{ \KwRet{\KwNone}}
    \Else{
      $d := d_{max}(|\SeqC|+1)|\SeqC|$\;
      $\AH :=$ \HBuild{$d,\SeqR,\SeqC$}\;
      \Switch{\Equiv{$\AH$}}{
        \lCase(\tcp*[f]{\color{blue}Found equivalent automaton}){\KwNone}{\KwRet{\KwSome{$\AH$}}} 
        \Case(\tcp*[f]{\color{blue}A counterexample $w$ has been found}){\KwSome{$w$}}{
          $(u,\sigma,c) :=$ \LPrefix{$\AH,w$}\;
          \lIf{$\rank{\AHankel{\SeqR \cdot [u]}{\SeqC}{d}} = |\SeqR|+1$}{ 
            \PLearner{$d_{max},\ell,\SeqR \cdot [u],\SeqC$}
          }\lElse{
            \PLearner{$d_{max},\ell,\SeqR \cdot [u],\SeqC \cdot [\sigma c]$}
          }
        }
      }
    }
  }
\caption{The partial learner}
\label{func:partial learner}
\end{algorithm}


\begin{proposition}
\label{lem:learner succeed}
Let $\A$ be the target automaton of dimension~$n$ and with $d_{max}$ the maximal degree of its polynomials.
Define $L(y_1,y_2):=\big((y_1+1)y_2^2\big)^{y_2}$.

Let $n' \geq n$ and $d \geq d_{max}$. Then 
\PLearner{$d,L(d,n'),[\varepsilon],[\varepsilon]$} $\neq \KwNone$.

\end{proposition}

\begin{proof}
Consider a totally suffix-closed sequence $\SeqC = [c_1,\ldots,c_m]$ of words. 
A simple induction on the length of the words~$w$ gives that the degree of the polynomials in $\backF{\A}(w)$ is at most $d_{max}(|w|+1)$. Since  $\SeqC$ is totally suffix-closed, we deduce that $c_1=\varepsilon$, as well as~$\max(|c_1|,\ldots,|c_m|) < m$. 
These two  above facts imply that the maximum degree of polynomials in $\backF{\A}(c_1)$, \ldots, $\backF{\A}(c_m)$ is at most~$d_{max}m$.

Recall that, by \Cref{cor:sequence of increasing modules}, every strictly increasing sequence   $\mathcal{M}_0 \subsetneq \ldots \subsetneq \mathcal{M}_k$ of submodules of~$\QQ[x]^{n}$ with the same rank $r$ has length~$k \leq d \cdot r$, where $d$ is the maximal degree of polynomials of the vectors generating $\mathcal{M}_0$.

For all $i \in \{1,\ldots,m\}$, define $\mathcal{M}_i := \mspan[{\QQ[x]}]{\backF{\A}(c_j) : j\leq i}$ .
The ranks of the submodules $\mathcal{M}_i$ are at most $n$. We aim at upper bounding~$m$ by $L(d_{\max},n)$; due to~\Cref{cor:sequence of increasing modules}, the worst upper bound is reached when some module in the  strictly increasing sequence of modules reaches full rank. Below we assume that that our increasing sequence reaches full rank, meaning that $\rank{\mathcal{M}_m} = n$.
Define $i_1,\ldots,i_n$ as the indices corresponding to when the rank of the submodules~$\mathcal{M}_{i}$ has strictly increased.
Formally, we have that $i_1=1$ and $i_1<\ldots<i_{n}$. Furthermore, for all $j\in \{2,\ldots,n\}$,
\[\rank{\mathcal{M}_{i_{j-1}}} \, =   \rank{\mathcal{M}_{i_j-1}} \, <\,  \rank{\mathcal{M}_{i_j}}\,.\]
By the above-mentioned bound on the degree of polynomials in the generators of~$\mathcal{M}_{i}$ together with \Cref{cor:sequence of increasing modules}, we infer the following properties: 
\begin{itemize}
    \item $m-i_n \, \leq \, \deg(\mathcal{M}_{i_n})n \, \leq \, i_n d_{max}n$;
    \item  and for all $j\in \{2,\cdots,n\}$, we have 
    \begin{equation}
    \label{eq:telescoping}
        (i_j-1) - i_{j-1} \leq \deg(\mathcal{M}_{i_{j-1}})  (j-1) \leq i_{j-1}  d_{max} (j-1)\,.
    \end{equation}
\end{itemize}

By telescoping~\eqref{eq:telescoping} from $i_j$ to~$i_1$, we get $i_j \leq j+ d_{max}\sum_{k=1}^{j-1} i_k k$, where the right-hand side
is at most~$j+d_{max} i_{j-1}\sum_{k=1}^{j-1} k \leq (d_{max}+1) i_{j-1}j^2$.
By a simple induction, for all $j \in \{2,\ldots,n\}$, we obtain that  $i_j \leq (d_{max}+1)^{j-1} j^{2(j-1)}$,   and so  $m \leq (d_{max}+1)^nn^{2n}=L(d_{max},n)$ holds .

Now we are ready to analyse the maximum  
number of  columns added to the sequence~$\SeqC$ during the successive recursive calls to the procedure \PLearner $\ell=L(d,n')$. 
Let $\SeqR_1,\, \SeqR_2, \, \ldots$ and $\SeqC_1, \, \SeqC_2,\, \ldots$ be  the successive values passed to \PLearner, where   $\SeqR_1, \SeqC_1$ are initialized to~$[\varepsilon]$.
 
 Using \Cref{lem:full rowrank and closeness}, the construction of the hypothesis automaton in \Cref{col:functionHbuild}, and by \Cref{cor:rank}, we deduce that for every $\SeqC_k$, in the successive values $\SeqC_1, \, \SeqC_2,\, \ldots$ as defined above, for~$d' := d(|\SeqC_k|+1)|\SeqC_k|$, the matrix $\AHankel{\SeqR_k}{\SeqC_k}{d'}$ is full row rank and also $\SeqC_k$ is totally  suffix-closed.
Moreover, writing  $\SeqC_k = [c_1,\ldots,c_k]$, as above, we define the sequence of  modules $\mathcal{M}_1, \, \ldots , \,\mathcal{M}_k$ 
where 
$\mathcal{M}_i = \mspan[{\QQ[x]}]{\backF{\A}(c_j) : j\leq i}$, with $1\leq i\leq k$.
By \Cref{lem:rank,lem:bound_column} we know that $\mathcal{M}_1 \subsetneq \ldots  \subsetneq\mathcal{M}_k$, that is, the sequence consists of strictly increasing modules. 
 Since $d \geq d_{max}$ and $n' \geq n$, the inequality 
 $L(n',d)\geq L(n,d_{\max}) $ implies that  $L(n',d) \geq |\SeqC_k|$ for every~$\SeqC_k$ in the sequence of $\SeqC_1, \, \SeqC_2,\, \ldots$. This concludes that \ELearner{$d,L(d,n'),[\varepsilon],[\varepsilon]$} $\neq \KwNone$.
\end{proof}


\subsubsection{Bounding the number of rows}
As previously mentioned, every row added by the learning algorithm is a prefix of a counterexample given by the teacher. In this section, we exhibit  a bound on the total number of rows added during the learning procedure, which is polynomial in the maximum size of the counterexamples and the target automaton. For this, we define a \emph{bounded forward vector space} that only considers words of bounded size.

Let $\A = (\init,\trans,\fin)$ be a P-finite automaton of dimension $n$ over $\Sigma$. Let $s \in \mathbb{N}$. The \emph{$s$-bounded forward function} associated with $\A$   is the function $\forFB{\A}{s}: \{ r \in \Sigma^* : |r| \leq s\} \rightarrow \QQ^{1\times (s+1)n}$ given by: 
\[
  \forFB{\A}{s}(u) = 
  \begin{bmatrix}
    \vec{0}_{1\times |u|n} & \init\trans(u,1) & \vec{0}_{1\times (s-|u|)n}
  \end{bmatrix}.
\]
The \emph{$s$-bounded forward space}, denoted $\forSB{\A}{s}$, is the vector space $\mspan[\QQ]{\forFB{\A}{s}(u) : |u| \leq s, u \in \Sigma^*}$.

Observe  that $\forSB{\A}{s}$ consists  of row vectors from $\QQ^{1\times (s+1)n}$. Thus, the dimension of the vector space~$\forSB{\A}{s}$  is at most $(s+1)n$. Intuitively, here $s$ represents  the maximal size of the counterexamples given by the teacher. In  the following proposition, we show that if $r_1,\ldots,r_m$ are the prefixes of the counterexamples added as rows  during the learning procedure, then the dimension of~$\mspan[\QQ]{\forFB{\A}{s}(r_1),\ldots,\forFB{\A}{s}(r_m)}$ is $m$.


\begin{proposition}
  \label{lem:bound_row}
  Let $\SeqR$ and $\SeqC$ be  sequences of words. Let $d,s \in \NN$ be such that $\AHankel{\SeqR}{\SeqC}{d}$ has full row rank, and  $|r|<s$ for all words~$r \in \SeqR$.  The  dimension of the vector space
$\mspan[\QQ]{\forFB{\A}{s}(r) : r \in \SeqR}$ is $|\SeqR|$.
\end{proposition}

\begin{proof}
Let $\SeqR = [r_1,\ldots,r_m]$. It suffices to argue that  the set $\{ \forFB{\A}{s}(r_i):1\leq i \leq m\}$ is linearly independent. 
Towards a contradiction, we assume  without loss of generality that 
$\forFB{\A}{s}(r_m)$ is dependant on the other $\forFB{\A}{s}(r_i)$, meaning that 
 that there exist~$q_1,\ldots,q_{m-1} \in \QQ$ such that $\forFB{\A}{s}(r_m) = \sum_{i=1}^{m-1} q_i \forFB{\A}{s}(r_i)$.
By the definition of $\forFB{\A}{s}$, as the vectors~$\forFB{\A}{s}(r_i)$ have only $n$ non-zero entries, we can further without loss of generality assume  that  $q_i \neq 0$, with $1\leq i\leq m-1$, implies $|r_i| = |r_m|$. Hence, for all $k \in \NN$,
\begin{align}
    (|r_m|+1)^k \forFB{\A}{s}(r_m) = \sum_{i=1}^{m-1} q_i (|r_m|+1)^k \forFB{\A}{s}(r_i) = \sum_{i=1}^{m-1} q_i (|r_i|+1)^k \forFB{\A}{s}(r_i). \label{eq:rmfb}
\end{align}

Finally, for all $c \in \Sigma^*$, denote by $B(c) \in \QQ^{n(s+1)\times 1}$ the column vector defined as follows:
\[
  B(c) = 
    \begin{bmatrix}
    \backF{\A}(c)(1)\\
    \vdots\\
    \backF{\A}(c)(s+1)
    \end{bmatrix}.
\]
By the definitions of forward and backward functions,  for all words $r \in \SeqR$ and $c \in \Sigma^*$, we have
\[\forFB{\A}{s}(r)B(c) = \init\trans(r,1)\backF{\A}(c)(|r|+1) = f(rc),\]
where the first equality holds due to the match of placement of non-zero entries in $\forFB{\A}{s}(r)$ with the placement of  $\backF{\A}(c)(|r|+1)$ in $B(c)$. 
 As an immediate result of the above equation and~\eqref{eq:rmfb}, 
we obtain that, for all words~$c \in \SeqC$ and for all $k \leq d$, 
\begin{align*}
  (|r_m|+1)^k f(r_mc) = & (|r_m|+1)^k \forFB{\A}{s}(r_m)B(c) \\
  = & \sum_{i=1}^{m-1} q_i (|r_i|+1)^k \forFB{\A}{s}(r_i)B(c) \\
  = & \sum_{i=1}^{m-1} q_i (|r_i|+1)^k f(r_ic).
\end{align*}
Therefore, if we denote by $A_1,\ldots,A_m$ the rows of $\AHankel{\SeqR}{\SeqC}{d}$ then we have shown that $A_m = \sum_{i=1}^{m-1} q_i A_i$ which is in contradiction with $\rank{\AHankel{\SeqR}{\SeqC}{d}} = |\SeqR|$.
\end{proof}


\begin{proposition}
  \label{lem:recursive call}
 Let $d,\ell \in \NN$. Let  $s$ be the maximal length of  counterexamples given by the teacher during the execution of \PLearner{$d,\ell,[\varepsilon],[\varepsilon]$}.
Then the number of recursive calls to \PLearner is at most~$(s+1)n$, where  $n$ is the number of states in the target automata.
\end{proposition}

\begin{proof}
Using \Cref{cor:rank} and \Cref{lem:full rowrank and closeness}, for every call to \PLearner with arguments $d,\ell,\SeqR,\SeqC$, we have $\rank{\AHankel{\SeqR}{\SeqC}{d'}} = |\SeqR|$ with $d' = d(|\SeqC|+1)|\SeqC|$. 

By \Cref{lem:bound_row}, the dimension of the vector space $\mspan[\QQ]{\forFB{\A}{s}(u) : u \in \SeqR}$ is $|\SeqR|$. But then, since $\mspan[\QQ]{\forFB{\A}{s}(u) : u \in \SeqR} \subseteq \QQ^{1\times (s+1)n}$, we have $|\SeqR| \leq (s+1)n$. 
We note that every recursive call to \PLearner increases the size of $\SeqR$ by one starting from $[\varepsilon]$. Therefore,  the number of recursive calls to \PLearner is at most $(s+1)n$.\end{proof}


\subsubsection{The exact learner}

The exact learner function is displayed in \Cref{func:exact learner}. The core of the learning process comes from the procedure~\PLearner. However, it still remains to correctly guess the values of $d_{max}$, the maximal degree of polynomials in the target automaton $\A$, and of $n$, its number of states. Guessing $n$ is important in order to compute the limit value $\ell$ for the number of columns added during the execution of \PLearner. As we need to guess two positive integers, we use the standard diagonal progression of the Cantor pairing function.

\begin{algorithm}[t]
  \Fn{\ELearner{}}{ 
    \SetAlgoVlined
    \SetAlgoShortEnd
    $sum := 1$\;
    \While{true}{
      \For{$n = 1$ to $sum$}{
        $d = sum - n$\;
        $\ell = 2(d+1)^nn^{2n}$\;
        \SetAlgoNoLine
        \SetAlgoNoEnd
        \Switch{\PLearner{$d,\ell,[\varepsilon],[\varepsilon]$}}{
          \lCase{\KwNone}{()}
          \lCase{\KwSome{$\AH$}}{\KwRet{$\AH$}}
        }
      }
      $sum := sum+1$\label{alg-label-sum}
    }
  }
\caption{The exact learner.}
\label{func:exact learner}
\end{algorithm}

The following theorem shows that  P-finite automata can be exactly learned in   time polynomial in the size of the target automaton and the maximal length of counterexamples given by the teacher.


\begin{theorem}
Let $\A$ be the target P-finite automaton. The procedure \ELearner terminates and returns a P-finite automaton $\AH$ such that $\sem{\A} = \sem{\AH}$. Moreover,  \ELearner runs in   time polynomial in the length of the encoding of~$\A$ and  the maximal length of counterexamples given by the teacher during the execution of \ELearner.
\end{theorem}

\begin{proof}
Let $n_{\A}$ be the number of states in the target automaton $\A$ and $d_{max}$ be the maximal degree of polynomials  in $\A$. 
Let $s$ be the maximal size of all counterexamples given by \Equiv during the execution of \ELearner.

 By \Cref{lem:recursive call}, when executing \PLearner{$d,\ell,[\varepsilon],[\varepsilon]$}, the number of recursive call to \PLearner is at most $(s+1)n_{\A}$, irrespective of the choice of~$d,\ell$.
 Furthermore, by the construction of \PLearner, we also know that for every call to \PLearner with arguments $d,\ell,\SeqR,\SeqC$, the size of $\SeqR$ is at most~$(s+1)n_{\A}$ and $|\SeqC| \leq |\SeqR|$.
 
 Define $d := d_{max}(|\SeqC|+1)|\SeqC|$. Observe that $d \leq d_{max}((s+1)n_{\A}+1)^2$. 
Recall that the procedure \HBuild{$d,\SeqR,\SeqC$} runs in time polynomial in $d + |\SeqR| + |\SeqC|$,  and also the procedure \LPrefix{$\AH,w$} runs in time  polynomial in $|w| + |\SeqC|$. We note that 
 the rank of $\AHankel{\SeqR \cdot [u]}{\SeqC}{d}$ can be computed in time polynomial in $(d+1)|\SeqC|(|\SeqR|+1)$ as well.

 By the above, the computation of one recursive call of \PLearner is polynomial in $s+d_{max}+n_{\A}$, which in turns implies  that each execution of \PLearner{$d,\ell,[\varepsilon],[\varepsilon]$} runs in time polynomial in $s$ and in the length of the encoding of~$\A$, irrespective of the choice of~$d,\ell$.

 Intuitively, 
 the variables~$n,d$ in  \ELearner 
 are the "guessed" values for the number of states and the degrees of polynomials in the target automaton, and 
 the variable~$sum$   is to implement the standard diagonal progression for~$n$ and $d$. 
By \Cref{lem:learner succeed}, we know that when the variable~$sum$ in Line~\ref{alg-label-sum} in \ELearner reaches the value $d_{max} + n_{\A}$, 
in the inner {\bf for}-loop, 
the output of  the call to \PLearner  with $d$  set to $ d_{max}$ and $\ell$  set to $(d+1)^n_{\A} n_{\A}^{2n_{\A}}$ is 
 necessarily different from \KwNone, terminating the  computation (the procedure might terminate before this point). Therefore, we will call \PLearner{$d,\ell,[\varepsilon],[\varepsilon]$} for different values of $d,\ell$ at most $(d_{max} + n_{\A})^2$ times, which  concludes the proof.
\end{proof}

%% file: app-same_rank.tex

\section{Proof of Proposition \ref{lemma:general_length}}
\label{sec:proof_seq_length}

Recall that in \Cref{sec:prelim} we considered a PID $R$, and $R$ being a PID implies that it is also a unique factorization domain. Thus, for all $r \in R$, we can compute the number of (not necessarily distinct) prime factors of $r$, that we denote $\pi(r)$, with the convention that $\pi(1_R) = 0$.

Given two $R$-modules $M$ and $N$, a \emph{homomorphism} between $M$ and $N$ is an $R$-linear map $\phi:M\rightarrow N$. That is, for all $r_1,r_2\in R$ and $m_1, m_2 \in M$, $\phi(r_1m_1+r_2m_2)=r_1\phi(m_1)+r_2\phi(m_2)$. When $\phi$ is additionally bijective, we say that $\phi$ is a \emph{isomorphism} and that $M$ and $N$ are \emph{isomorphic}, denoted $M \simeq N$. 

We also define the \emph{direct-sum} $R$-module as $M \oplus N = \{ (m,n) : m \in M, n\in N\}$. We say that $M$ is \emph{torsion} if for all $m \in M$, there exists $r \in R \setminus \{0\}$ such that $rm = 0$.
When $N$ is a submodule of $M$, the \emph{quotient} $R$-module $M/N$ is the set of elements of the form $m+N = \{ m+n : n \in N\}$ for some $m \in M$, endowed with addition and multiplication operations as follows: for all $r \in R$, and $m,m' \in M$, we have $(m+N) + (m'+N) = (m+m') + N$ and $r(m+N) = (rm)+N$.  

For example, $3\ZZ$ is a $\ZZ$-submodule of $\ZZ$ and the quotient $\ZZ/3\ZZ$ is a torsion $\ZZ$-module that contains the elements $0 + 3\ZZ$, $1+3\ZZ$ and $2+3\ZZ$, isomorphic to $\ZZ_3$.

As mentioned in \Cref{sec:prelim}, for any submodules $M$ of $R^n$, there exist an $R$-basis $f_1,\ldots,f_n$ of $R^n$ and elements $d_1,\ldots,d_r \in R$ such that $d_1f_1,\ldots,d_rf_r$ is an $R$-basis of $M$. Moreover, when $M$ is generated by some vectors $\vv_1,\ldots,\vv_m$, each $d_i$ are defined as $\frac{D_i(A)}{D_{i-1}(A)}$ where $A \in R^{n\times m}$ is the matrix whose columns are the vectors $\vv_1,\ldots,\vv_m$ with $D_0(A) = 1$ by convention. This correlates directly with the following theorem.


\begin{theorem}[{\cite[Theorem 6.12]{Hungerford74}}]
Let $M$ be a finitely generated $R$-module. There exists $d_1,\ldots,d_r \in R$ such that $M$ is isomorphic to $R^k \oplus \bigoplus_{j=1}^r R/\mspan[R]{d_j}$ for some $k$ and \emph{invariant factors} $d_1 \mid d_2 \mid \ldots \mid d_r$ that are unique (up to multiplication by units). Moreover, $M$ is torsion iff $k=0$, and in such a case we define the \emph{$R$-dimension} of $M$, denoted $\dim_R(M)$, as $\pi(d_1d_2\ldots d_r)$. 
\end{theorem}

The proof of \Cref{lemma:general_length} relies on the following well known result:


\begin{theorem}[{\cite[Proposition 1]{cassels2010algebraic}}]
\label{thm:index_decomposition}
Let $M,M',M''$ be $R$-submodules of $R^n$ such that $M''\subset M'\subset M$, all having the same rank. Then \[\dim_R(M/M'') = \dim_R(M/M') + \dim_R(M'/M'').\]
\end{theorem}


The $R$-dimensions of $M/M''$, $M/M'$, and $M'/M''$ are well defined as one can observe that two finitely generated modules have the same rank if and only if their quotient module is torsion.
We are now ready to prove a general bound on the length of a strictly increasing sequence of submodules of the same rank.

\lemmagenerallength*

\begin{proof}
  Let $f_1,\ldots,f_n$ be a $R$-basis of $R^n$ and let $d_1,\ldots,d_r \in R$ such that $d_1f_1,\ldots,d_rf_r$ is a $R$-basis of $M$ and $d_1\mid d_2 \mid \ldots \mid d_r$. Recall that for all $i \in \{1,\ldots,r\}$, $d_i = \frac{D_i(A)}{D_{i-1}(A)}$ and $D_0(A) = 1_R$ by convention.
    
  Let $N= \mspan[R]{f_1,\ldots, f_r }$. We first show that $M_k\subseteq N$, so that the whole increasing sequence is in $N$. Assume, for a contradiction, that there is some $m\in M_k$ but $m\notin N$. Then, when expressing $m$ in the basis $f_1,\ldots,f_n$, at least one coefficient of some $f_i$, $i>r$, must be non-zero. This further implies that the vectors $m$, $d_1f_1$, $d_2f_2$,\ldots, $d_rf_r$ are linearly independent (as some linear dependence would result in a linear dependence between $f_i,f_1,f_2,\ldots,f_r$). But, as $\{d_1f_1,\ldots,d_rf_r\}\subseteq M\subseteq M_k$, the module $M_k$ contains $r+1$ linearly independent elements. Their span is contained in $M_k$ and has rank $r+1$, which contradicts the fact that $M_k$ has rank $r$. 

  We are now ready to inductively apply Theorem \ref{thm:index_decomposition} to the chain $M\subsetneq M_1\subsetneq M_2 \subsetneq \ldots \subsetneq M_k\subseteq N$. Note that all these submodules have rank $r$, as $M$ does (by definition), and so does $N$ (by construction). We have:
  \[\dim_R(N/M) = \dim_R(N/M_k)+\dim_R(M_k/M_{k-1})+\ldots+\dim_R(M_2/M_1)+\dim_R(M_1/M).\]
  The first term on the right-hand side is clearly non-negative, so we can bound $k$ as
  \[k\leq \frac{\dim_R(N/M)}{\underset{1\leq i \leq k}{\textrm{min}}\:\dim_R(M_{i}/M_{i-1})}.\]

  For all $i\in\{1,\ldots,k\}$, we can show that $\dim_R(M_i/M_{i-1})\geq 1$. Otherwise, by contradiction, there would exists $i \in \{1,\ldots,k\}$ such that $\dim_R(M_i/M_{i-1})=0$, which entails $M_i/M_{i-1}$ is isomorphic to $\bigoplus_{k=1}^{r'} R/\mspan[R]{1_R}$ for some $r'$. As $R/\mspan[R]{1_R}$ is isomorphic to $\{0\}$ then so is $M_i/M_{i-1}$. This implies that $M_i = M_{i-1}$, contradicting our hypothesis that $M_{i-1} \subsetneq M_i$.
  
  Finally, observe that $\dim_R(N/M) = \pi(d_1d_2\ldots d_r)$ since $d_1f_1,\ldots,d_rf_r$ is a basis of $M$ and $N = \mspan[R]{f_1,\ldots,f_r}$. As $d_i = \frac{D_i(A)}{D_{i-1}(A)}$ for all $i \in \{1,\ldots,r\}$, we conclude that $\pi(d_1\ldots d_r) = \pi(\frac{D_r(A)}{D_1(A)}) = \pi(D_r(A))$ and so $k \leq \pi(D_r(A))$.
\end{proof}

%% file: app-holonomic.tex
\section{Proofs of Section~\ref{subsec:Hequivalence}}
\label{sec:proof holonomic equiv}

\lemmareductionequivtozero*

\begin{proof}
  Let  $\A_i=\left(\init_i,\trans_i,\fin_i(x)\right)$ for $i\in\{1,2\}$  be two P-finite automata over $\Sigma$, respectively, of dimension $n_1$ and $n_2$.  We construct a  P-finite automata $\mathcal{A}_-=\left(\init,\trans,\fin(x)\right)$ of dimension~$n_1+n_2$ over $\Sigma$ such that
  \[\init=\begin{bmatrix}
      \init_1 & -\init_2
      \end{bmatrix} \qquad \text{ and } \qquad \fin(x)=\begin{bmatrix}
      \fin_1(x) \\ \fin_2(x)
      \end{bmatrix},\]
and  for all $\sigma\in\Sigma$,
\[\trans(\sigma,x)=\begin{bmatrix}      \trans_1(\sigma,x) & 0_{n_1\times n_2}\\
      0_{n_2\times n_1} & \trans_2(\sigma,x)
      \end{bmatrix}\,.\]
The  construction of~$\A$ can be done in time polynomial in the size of encoding of~$\A_1$ and $\A_2$.

We first claim  that  $f_{\mathcal{A}_{-}}(w)=f_{\A_{1}}(w)-f_{\A_{2}}(w)$ for all $w\in\Sigma^*$.  
  By definition, 
  \[f_{\A_-}(\varepsilon) = \init \fin(1) = \init_1\fin_1(1) - \init_2\fin_2(1) = f_{\A_1}(\varepsilon) - f_{\A_2}(\varepsilon)\,,\]
  as required. 
  Consider $w = \sigma_1\ldots \sigma_k$ with the $\sigma_i \in \Sigma$.  We have that 
  \begin{align*}
      f_{\A_-}(w)&=\begin{bmatrix}
      \init_1 & -\init_2
      \end{bmatrix}\begin{bmatrix}
      \trans_1(\sigma_1,1) & 0_{n_1\times n_2}\\
      0_{n_2\times n_1} & \trans_2(\sigma_1,1)
      \end{bmatrix}\cdots\begin{bmatrix}
      \trans_1(\sigma_k,k) & 0_{n_1\times n_2}\\
      0_{n_2\times n_1} & \trans_2(\sigma_k,k)
      \end{bmatrix}
      \begin{bmatrix}
      \fin_1(k+1) \\ \fin_2(k+1)
      \end{bmatrix}\\
      &=\begin{bmatrix}
      \init_1 & -\init_2
      \end{bmatrix}\begin{bmatrix}
      \trans_1(\sigma_1,1)\ldots \trans_1(\sigma_k,k) & 0_{n_1\times n_2}\\
      0_{n_2\times n_1} & \trans_2(\sigma_1,1)\ldots \trans_2(\sigma_k,k)
      \end{bmatrix}
      \begin{bmatrix}
      \fin_1(k+1) \\ \fin_2(k+1)
      \end{bmatrix}\\
      &=\init_1 \trans_1(\sigma_1,1)\ldots \trans_1(\sigma_k,k)\fin_1(k+1)-\init_2 \trans_2(\sigma_1,1)\ldots \trans_2(\sigma_k,k)\fin_2(k+1)\\
      &=f_{\A_1}(w)-f_{\A_2}(w).\qedhere
  \end{align*}

We  can now conclude the proof
by observing that $\sem{\A}$ is identically zero if and only if $\sem{\A_1}=\sem{\A_2}$, as required.
  \end{proof}

%% file: main.bbl

\begin{thebibliography}{27}


\ifx \showCODEN    \undefined \def \showCODEN     #1{\unskip}     \fi
\ifx \showDOI      \undefined \def \showDOI       #1{#1}\fi
\ifx \showISBNx    \undefined \def \showISBNx     #1{\unskip}     \fi
\ifx \showISBNxiii \undefined \def \showISBNxiii  #1{\unskip}     \fi
\ifx \showISSN     \undefined \def \showISSN      #1{\unskip}     \fi
\ifx \showLCCN     \undefined \def \showLCCN      #1{\unskip}     \fi
\ifx \shownote     \undefined \def \shownote      #1{#1}          \fi
\ifx \showarticletitle \undefined \def \showarticletitle #1{#1}   \fi
\ifx \showURL      \undefined \def \showURL       {\relax}        \fi
\providecommand\bibfield[2]{#2}
\providecommand\bibinfo[2]{#2}
\providecommand\natexlab[1]{#1}
\providecommand\showeprint[2][]{arXiv:#2}

\bibitem[Angluin(1987)]%
        {Angluin87}
\bibfield{author}{\bibinfo{person}{Dana Angluin}.} \bibinfo{year}{1987}\natexlab{}.
\newblock \showarticletitle{Learning Regular Sets from Queries and Counterexamples}.
\newblock \bibinfo{journal}{\emph{Inf. Comput.}} \bibinfo{volume}{75}, \bibinfo{number}{2} (\bibinfo{year}{1987}), \bibinfo{pages}{87--106}.
\newblock
\urldef\tempurl%
\url{https://doi.org/10.1016/0890-5401(87)90052-6}
\showDOI{\tempurl}


\bibitem[Angluin and Kharitonov(1995)]%
        {AngluinK95}
\bibfield{author}{\bibinfo{person}{Dana Angluin} {and} \bibinfo{person}{Michael Kharitonov}.} \bibinfo{year}{1995}\natexlab{}.
\newblock \showarticletitle{When Won't Membership Queries Help?}
\newblock \bibinfo{journal}{\emph{J. Comput. Syst. Sci.}} \bibinfo{volume}{50}, \bibinfo{number}{2} (\bibinfo{year}{1995}), \bibinfo{pages}{336--355}.
\newblock


\bibitem[Beimel et~al\mbox{.}(1999)]%
        {exact-learning-wa}
\bibfield{author}{\bibinfo{person}{Amos Beimel}, \bibinfo{person}{Francesco Bergadano}, \bibinfo{person}{Nader Bshouty}, \bibinfo{person}{Eyal Kushilevitz}, {and} \bibinfo{person}{Stefano Varricchio}.} \bibinfo{year}{1999}\natexlab{}.
\newblock \showarticletitle{Learning Functions Represented as Multiplicity Automata}.
\newblock \bibinfo{journal}{\emph{J. ACM}}  \bibinfo{volume}{47} (\bibinfo{date}{10} \bibinfo{year}{1999}).
\newblock
\urldef\tempurl%
\url{https://doi.org/10.1007/978-0-387-30162-4_194}
\showDOI{\tempurl}


\bibitem[Benedikt et~al\mbox{.}(2017)]%
        {BenediktDSW17}
\bibfield{author}{\bibinfo{person}{Michael Benedikt}, \bibinfo{person}{Timothy Duff}, \bibinfo{person}{Aditya Sharad}, {and} \bibinfo{person}{James Worrell}.} \bibinfo{year}{2017}\natexlab{}.
\newblock \showarticletitle{Polynomial automata: Zeroness and applications}. In \bibinfo{booktitle}{\emph{32nd Annual {ACM/IEEE} Symposium on Logic in Computer Science, {LICS} 2017}}. \bibinfo{publisher}{{IEEE} Computer Society}, \bibinfo{pages}{1--12}.
\newblock


\bibitem[Berstel and Reutenauer(1988)]%
        {berstel1988rational}
\bibfield{author}{\bibinfo{person}{Jean Berstel} {and} \bibinfo{person}{Christophe Reutenauer}.} \bibinfo{year}{1988}\natexlab{}.
\newblock \bibinfo{booktitle}{\emph{Rational series and their languages}}. Vol.~\bibinfo{volume}{12}.
\newblock \bibinfo{publisher}{Springer-Verlag}.
\newblock


\bibitem[Berstel and Reutenauer(2010)]%
        {BerstelR10}
\bibfield{author}{\bibinfo{person}{Jean Berstel} {and} \bibinfo{person}{Christophe Reutenauer}.} \bibinfo{year}{2010}\natexlab{}.
\newblock \bibinfo{booktitle}{}. \bibinfo{series}{Encyclopedia of Mathematics and its Applications}, Vol.~\bibinfo{volume}{137}.
\newblock \bibinfo{publisher}{Cambridge University Press}.
\newblock


\bibitem[Bollig et~al\mbox{.}(2009)]%
        {BolligHKL09}
\bibfield{author}{\bibinfo{person}{Benedikt Bollig}, \bibinfo{person}{Peter Habermehl}, \bibinfo{person}{Carsten Kern}, {and} \bibinfo{person}{Martin Leucker}.} \bibinfo{year}{2009}\natexlab{}.
\newblock \showarticletitle{Angluin-Style Learning of {NFA}}. In \bibinfo{booktitle}{\emph{{IJCAI} 2009, Proceedings of the 21st International Joint Conference on Artificial Intelligence}}. \bibinfo{pages}{1004--1009}.
\newblock


\bibitem[Bollig et~al\mbox{.}(2010)]%
        {BolligKKLNP10}
\bibfield{author}{\bibinfo{person}{Benedikt Bollig}, \bibinfo{person}{Joost{-}Pieter Katoen}, \bibinfo{person}{Carsten Kern}, \bibinfo{person}{Martin Leucker}, \bibinfo{person}{Daniel Neider}, {and} \bibinfo{person}{David~R. Piegdon}.} \bibinfo{year}{2010}\natexlab{}.
\newblock \showarticletitle{libalf: The Automata Learning Framework}. In \bibinfo{booktitle}{\emph{{CAV} 2010}} \emph{(\bibinfo{series}{LNCS}, Vol.~\bibinfo{volume}{6174})}. \bibinfo{publisher}{Springer}, \bibinfo{address}{Edinburgh, UK}, \bibinfo{pages}{360--364}.
\newblock


\bibitem[Cassels and Fr{\"o}hlich(2010)]%
        {cassels2010algebraic}
\bibfield{author}{\bibinfo{person}{J.W.S. Cassels} {and} \bibinfo{person}{A. Fr{\"o}hlich}.} \bibinfo{year}{2010}\natexlab{}.
\newblock \bibinfo{booktitle}{\emph{Algebraic Number Theory: Proceedings of an Instructional Conference Organized by the London Mathematical Society (a NATO Advanced Study Institute) with the Support of the International Mathematical Union}}.
\newblock \bibinfo{publisher}{London Mathematical Society}.
\newblock
\showISBNx{9780950273426}
\showLCCN{2010497874}
\urldef\tempurl%
\url{https://books.google.fr/books?id=DQP\_RAAACAAJ}
\showURL{%
\tempurl}


\bibitem[Fliess(1974)]%
        {fliess1974matrices}
\bibfield{author}{\bibinfo{person}{Michel Fliess}.} \bibinfo{year}{1974}\natexlab{}.
\newblock \showarticletitle{Matrices de hankel}.
\newblock \bibinfo{journal}{\emph{J. Math. Pures Appl}} \bibinfo{volume}{53}, \bibinfo{number}{9} (\bibinfo{year}{1974}), \bibinfo{pages}{197--222}.
\newblock


\bibitem[Howar et~al\mbox{.}(2019)]%
        {HowarJV19}
\bibfield{author}{\bibinfo{person}{Falk Howar}, \bibinfo{person}{Bengt Jonsson}, {and} \bibinfo{person}{Frits~W. Vaandrager}.} \bibinfo{year}{2019}\natexlab{}.
\newblock \showarticletitle{Combining Black-Box and White-Box Techniques for Learning Register Automata}.
\newblock In \bibinfo{booktitle}{\emph{Computing and Software Science - State of the Art and Perspectives}}. \bibinfo{series}{LNCS}, Vol.~\bibinfo{volume}{10000}. \bibinfo{publisher}{Springer}, \bibinfo{address}{Cham}, \bibinfo{pages}{563--588}.
\newblock


\bibitem[Humenberger et~al\mbox{.}(2017a)]%
        {humenberger2017}
\bibfield{author}{\bibinfo{person}{Andreas Humenberger}, \bibinfo{person}{Maximilian Jaroschek}, {and} \bibinfo{person}{Laura Kov{\'a}cs}.} \bibinfo{year}{2017}\natexlab{a}.
\newblock \showarticletitle{Automated generation of non-linear loop invariants utilizing hypergeometric sequences}. In \bibinfo{booktitle}{\emph{Proceedings of the 2017 ACM on International Symposium on Symbolic and Algebraic Computation}}. \bibinfo{pages}{221--228}.
\newblock


\bibitem[Humenberger et~al\mbox{.}(2017b)]%
        {humenberger2017invariant}
\bibfield{author}{\bibinfo{person}{Andreas Humenberger}, \bibinfo{person}{Maximilian Jaroschek}, {and} \bibinfo{person}{Laura Kov{\'a}cs}.} \bibinfo{year}{2017}\natexlab{b}.
\newblock \showarticletitle{Invariant generation for multi-path loops with polynomial assignments}. In \bibinfo{booktitle}{\emph{International Conference on Verification, Model Checking, and Abstract Interpretation}}. Springer, \bibinfo{pages}{226--246}.
\newblock


\bibitem[Hungerford(1974)]%
        {Hungerford74}
\bibfield{author}{\bibinfo{person}{Thomas~W. Hungerford}.} \bibinfo{year}{1974}\natexlab{}.
\newblock \bibinfo{booktitle}{\emph{Algebra}}.
\newblock \bibinfo{publisher}{Springer}, \bibinfo{address}{New York, NY}, \bibinfo{pages}{225}.
\newblock


\bibitem[Isberner et~al\mbox{.}(2015)]%
        {IsbernerHS15}
\bibfield{author}{\bibinfo{person}{Malte Isberner}, \bibinfo{person}{Falk Howar}, {and} \bibinfo{person}{Bernhard Steffen}.} \bibinfo{year}{2015}\natexlab{}.
\newblock \showarticletitle{The Open-Source LearnLib - {A} Framework for Active Automata Learning}. In \bibinfo{booktitle}{\emph{{CAV} 2015}} \emph{(\bibinfo{series}{LNCS}, Vol.~\bibinfo{volume}{9206})}. \bibinfo{publisher}{Springer}, \bibinfo{address}{San Francisco, CA, USA}, \bibinfo{pages}{487--495}.
\newblock


\bibitem[Kannan(1985)]%
        {KANNAN198569}
\bibfield{author}{\bibinfo{person}{R. Kannan}.} \bibinfo{year}{1985}\natexlab{}.
\newblock \showarticletitle{Solving systems of linear equations over polynomials}.
\newblock \bibinfo{journal}{\emph{Theoretical Computer Science}}  \bibinfo{volume}{39} (\bibinfo{year}{1985}), \bibinfo{pages}{69--88}.
\newblock
\showISSN{0304-3975}
\urldef\tempurl%
\url{https://doi.org/10.1016/0304-3975(85)90131-8}
\showDOI{\tempurl}
\newblock
\shownote{Third Conference on Foundations of Software Technology and Theoretical Computer Science}.


\bibitem[Kauers and Paule(2011)]%
        {Tetrahedron}
\bibfield{author}{\bibinfo{person}{Manuel Kauers} {and} \bibinfo{person}{Peter Paule}.} \bibinfo{year}{2011}\natexlab{}.
\newblock \bibinfo{booktitle}{\emph{{The Concrete Tetrahedron}} (\bibinfo{edition}{1st} ed.)}.
\newblock \bibinfo{publisher}{Springer Wien}.
\newblock


\bibitem[Kiefer(2020)]%
        {kiefer2020notes}
\bibfield{author}{\bibinfo{person}{Stefan Kiefer}.} \bibinfo{year}{2020}\natexlab{}.
\newblock \bibinfo{title}{Notes on Equivalence and Minimization of Weighted Automata}.
\newblock
\newblock
\showeprint[arxiv]{2009.01217}~[cs.FL]


\bibitem[Kov{\'{a}}cs(2008)]%
        {Kovacs08}
\bibfield{author}{\bibinfo{person}{Laura Kov{\'{a}}cs}.} \bibinfo{year}{2008}\natexlab{}.
\newblock \showarticletitle{Reasoning Algebraically About P-Solvable Loops}. In \bibinfo{booktitle}{\emph{Tools and Algorithms for the Construction and Analysis of Systems, 14th International Conference, {TACAS}, Proceedings}} \emph{(\bibinfo{series}{Lecture Notes in Computer Science}, Vol.~\bibinfo{volume}{4963})}. \bibinfo{publisher}{Springer}, \bibinfo{pages}{249--264}.
\newblock


\bibitem[Michaliszyn and Otop(2022)]%
        {MichaliszynO22}
\bibfield{author}{\bibinfo{person}{Jakub Michaliszyn} {and} \bibinfo{person}{Jan Otop}.} \bibinfo{year}{2022}\natexlab{}.
\newblock \showarticletitle{Learning Deterministic Visibly Pushdown Automata Under Accessible Stack}. In \bibinfo{booktitle}{\emph{47th International Symposium on Mathematical Foundations of Computer Science, {MFCS}}} \emph{(\bibinfo{series}{LIPIcs}, Vol.~\bibinfo{volume}{241})}. \bibinfo{publisher}{Schloss Dagstuhl - Leibniz-Zentrum f{\"{u}}r Informatik}, \bibinfo{pages}{74:1--74:16}.
\newblock


\bibitem[Moerman et~al\mbox{.}(2017)]%
        {MoermanS0KS17}
\bibfield{author}{\bibinfo{person}{Joshua Moerman}, \bibinfo{person}{Matteo Sammartino}, \bibinfo{person}{Alexandra Silva}, \bibinfo{person}{Bartek Klin}, {and} \bibinfo{person}{Michal Szynwelski}.} \bibinfo{year}{2017}\natexlab{}.
\newblock \showarticletitle{Learning nominal automata}. In \bibinfo{booktitle}{\emph{Proceedings of the 44th {ACM} {SIGPLAN} Symposium on Principles of Programming Languages, {POPL}}}. \bibinfo{publisher}{{ACM}}, \bibinfo{pages}{613--625}.
\newblock


\bibitem[Newman(1997)]%
        {newman1997smith}
\bibfield{author}{\bibinfo{person}{Morris Newman}.} \bibinfo{year}{1997}\natexlab{}.
\newblock \showarticletitle{The Smith normal form}.
\newblock \bibinfo{journal}{\emph{Linear algebra and its applications}} \bibinfo{volume}{254}, \bibinfo{number}{1-3} (\bibinfo{year}{1997}), \bibinfo{pages}{367--381}.
\newblock


\bibitem[Reutenauer(2012)]%
        {Reutenauer12}
\bibfield{author}{\bibinfo{person}{Christophe Reutenauer}.} \bibinfo{year}{2012}\natexlab{}.
\newblock \showarticletitle{On a Matrix Representation for Polynomially Recursive Sequences}.
\newblock \bibinfo{journal}{\emph{Electron. J. Comb.}} \bibinfo{volume}{19}, \bibinfo{number}{3} (\bibinfo{year}{2012}), \bibinfo{pages}{36}.
\newblock


\bibitem[Sch{\"{u}}tzenberger(1961)]%
        {Schutzenberger61b}
\bibfield{author}{\bibinfo{person}{Marcel~Paul Sch{\"{u}}tzenberger}.} \bibinfo{year}{1961}\natexlab{}.
\newblock \showarticletitle{On the Definition of a Family of Automata}.
\newblock \bibinfo{journal}{\emph{Inf. Control.}} \bibinfo{volume}{4}, \bibinfo{number}{2-3} (\bibinfo{year}{1961}), \bibinfo{pages}{245--270}.
\newblock


\bibitem[Smith(1861)]%
        {SNF}
\bibfield{author}{\bibinfo{person}{Henry John~Stephen Smith}.} \bibinfo{year}{1861}\natexlab{}.
\newblock \showarticletitle{{XV}. On systems of linear indeterminate equations and congruences}.
\newblock \bibinfo{journal}{\emph{Philosophical Transactions of the Royal Society of London}}  \bibinfo{volume}{151} (\bibinfo{date}{Dec.} \bibinfo{year}{1861}), \bibinfo{pages}{293--326}.
\newblock
\urldef\tempurl%
\url{https://doi.org/10.1098/rstl.1861.0016}
\showDOI{\tempurl}


\bibitem[Tzeng(1992)]%
        {doi:10.1137/0221017}
\bibfield{author}{\bibinfo{person}{Wen-Guey Tzeng}.} \bibinfo{year}{1992}\natexlab{}.
\newblock \showarticletitle{A Polynomial-Time Algorithm for the Equivalence of Probabilistic Automata}.
\newblock \bibinfo{journal}{\emph{SIAM J. Comput.}} \bibinfo{volume}{21}, \bibinfo{number}{2} (\bibinfo{year}{1992}), \bibinfo{pages}{216--227}.
\newblock
\urldef\tempurl%
\url{https://doi.org/10.1137/0221017}
\showDOI{\tempurl}


\bibitem[van Heerdt et~al\mbox{.}(2020)]%
        {HeerdtKR020}
\bibfield{author}{\bibinfo{person}{Gerco van Heerdt}, \bibinfo{person}{Clemens Kupke}, \bibinfo{person}{Jurriaan Rot}, {and} \bibinfo{person}{Alexandra Silva}.} \bibinfo{year}{2020}\natexlab{}.
\newblock \showarticletitle{Learning Weighted Automata over Principal Ideal Domains}. In \bibinfo{booktitle}{\emph{Foundations of Software Science and Computation Structures - 23rd International Conference, {FOSSACS} 2020, Proceedings}} \emph{(\bibinfo{series}{Lecture Notes in Computer Science}, Vol.~\bibinfo{volume}{12077})}. \bibinfo{publisher}{Springer}, \bibinfo{pages}{602--621}.
\newblock


\end{thebibliography}
